\documentclass[aps,pre,twocolumn,superscriptaddr+ess,showpacs,showkeys]{revtex4-2}
\usepackage[english]{babel}
\usepackage[utf8]{inputenc}
\usepackage[colorinlistoftodos, color=green!40, prependcaption]{todonotes}
\usepackage{amsthm}
\usepackage{mathtools}
\usepackage{physics}
\usepackage{xcolor}
\usepackage{graphicx}
\usepackage[left=23mm,right=13mm,top=35mm,columnsep=15pt]{geometry} 
\usepackage{adjustbox}
\usepackage{placeins}
\usepackage[T1]{fontenc}
\usepackage{lipsum}
\usepackage{csquotes}
\usepackage[pdftex, pdftitle={Article}, pdfauthor={Author}]{hyperref} % For hyperlinks in the PDF
\usepackage{braket}
\usepackage{amssymb}
\usepackage{booktabs}
\usepackage{hyperref}
\theoremstyle{definition}

\newtheorem{theorem}{Theorem}
\usepackage{array,tabularx}
\usepackage{makecell}
\usepackage{soul}
\newcolumntype{P}[1]{>{\centering\arraybackslash}p{#1}}
\definecolor{redcolor}{rgb}{1,0,0}

\definecolor{bluecolor}{rgb}{0,0,1}

\begin{document}
\title{Closed-form solutions for the Salpeter equation}

\author{Fernando Alonso-Marroquin}
\affiliation{CIPR, King Fahd University of Petroleum and Minerals, Dhahran, 31261, Kingdom of Saudi Arabia}
\email[Corresponding Email: ]{fernando.marroquin@kfupm.edu.sa}% Your name
	
\author{Yaoyue Tang}
\affiliation{Modelling and Simulation Research Group, The University of Sydney, Sydney NSW 2006, Australia} 	

 \author{Fatemeh Gharari}
\affiliation{Department of Statistics and Computer Science, University of Mohaghegh Ardabili, Ardabil, Iran}

\author{M. N. Najafi}
\affiliation{Department of Physics, University of Mohaghegh Ardabili, P.O. Box 179, Ardabil, Iran}
\email{morteza.nattagh@gmail.com}

\begin{abstract}
We propose integral representations and analytical solutions for the propagator of the $1+1$ dimensional Salpeter Hamiltonian, describing a relativistic quantum particle with no spin. We explore the exact Green function and an exact solution for a given initial condition, and also find the asymptotic solutions in some limiting cases. The analytical extension of the Hamiltonian in the complex plane allows us to formulate the equivalent stochastic problem, namely the B\"aumer equation. This equation describes \textit{relativistic} stochastic processes with time-changing anomalous diffusion.  This B\"aumer propagator corresponds to the Green function of a relativistic diffusion process that interpolates between Cauchy distributions for small times and Gaussian diffusion for large times, providing a framework for stochastic processes where anomalous diffusion is time-dependent.
\end{abstract}


\maketitle

\section{introduction}\label{sec:introduction}
The importance of the relativistic quantum mechanics (RQM) is not only due to the requirements of the standard model of elementary particles~\cite{greiner2000relativistic}, but also as an essential ingredient for describing the new phases of matter, like two-dimensional electron gas~\cite{melrose1984dispersion}, graphene~\cite{novoselov2005two}, relativistic plasmas~\cite{umstadter2001review,umstadter1996nonlinear,arnold2002photon}, neutron stars~\cite{watts2016colloquium}, Weyl semi-metals~\cite{lv2015experimental}, relativistic hydrodynamics~\cite{rezzolla2013relativistic}, and other materials with Dirac cone dispersion in their band structure~\cite{malko2012two}. RQM has vast applications in other fields, such as the stochastic phenomena which is realized using a time Wick's rotation. The examples are the relativistic stochastic phenomena~\cite{dunkel2009relativistic,debbasch2007relativistic,de1986stochastic,deng2009finite}, the relativistic Ornstein-Uhlenbeck process~\cite{debbasch1997relativistic,rigotti2005relativistic}, the relativistic Brownian motion~\cite{pal2020stochastic,kakushadze2017volatility}, the relativistic Fokker-Planck equation~\cite{hakim1968relativistic,kleinert2013green,serva1988relativistic}, the classical and relativistic Gross-Pitaevskii equation~\cite{kleinert2012fractional} and the relativistic diffusion~\cite{chevalier2008relativistic}. Due to its wide applications, the standard theory for relativistic (classical and quantum) statistical mechanics is nowadays well-developed~\cite{mi2011introduction}. There are different ways in order to construct a relativistic version of classical quantum mechanics that respect relativity due to their covariant formulation. The Klein-Gordon equation~\cite{greiner2000relativistic} containing second time derivative, the Dirac equations~\cite{thaller2013dirac} which is of the first order in the time derivative, and the Salpeter equation~\cite{foldy1956synthesis,kowalski2011salpeter}, which is known as the ``square root'' of the Klein-Gordon operator are the famous formulations to this end. The Salpeter equation governs scalar particles since its relativistic covariance is applied at the operator level, while the Dirac equation and the Klein-Gordon equation are capable of describing higher dimensional representations of the Lorenz group. Although working with a square root of the Klein-Gordon operator makes the Salpeter equation difficult to handle, the fact that it resolves some essential problems with the Klein-Gordon equation like the negative energy particles is a privilege, making it worth putting under intense studies~\cite{kowalski2011salpeter}. The square-root operator induces non-locality for Salpeter's equation, a fact that makes it troublesome. An analysis for the conserved current, and also second quantization formulation can be found in~\cite{lammerzahl1993pseudodifferential}. Despite its importance, there are limited works that focus on its structure and consequences. For example, the causality structure of this model, its plane wave solutions, the superposition of quantum states, entanglement structure, second quantized field theoretical description, and the corresponding particle algebra that retrieves the causality, the closed form of the green function, and the dual stochastic theory (after the Wick rotation) is missing in the literature. The aim of the present paper is to process the later gap, i.e., the closed form of the propagator of the free particles, and the dual stochastic interpretation.  

A motivation for the Salpeter stochastic description is the crossover phenomena in diffusion processes that takes place for example in the analysis of stochastic processes, such as the evolution of the probability density function of the S\&P 500 stock market index \cite{alonso2019q}. This crossover is manifested in a superdiffusive regime for small times and a Gaussian diffusive regime for large times. Attempts have been made to use quantum mechanic models to measure stock return volatility, and reported a power-law relation in the tails of price return distribution \cite{ahn2024business}. For small momentum values, the Salpeter equation is analogous to the Schr\"odinger equation (where in the dual stochastic equation, the solution is normal diffusion), while in the large momentum it corresponds to massless free particles \cite{kowalski2011salpeter}. On the other hand, the Green function of the stochastic equation is related to the Cauchy distributions \cite{baeumer2010stochastic}. Therefore, we expect a crossover between these two limits for intermediate momenta, which induces a crossover in terms of the spatiotemporal coordinates. 

This paper presents analytical solutions of the Salpeter propagator in its quantum and stochastic versions. 
The paper is organized as follows: In the next section, we introduce the Salpeter equation to be analyzed in this paper and provide the integral representation and analytical solutions of the propagator.  In Sec.\ref{SEC:RelStoch} we present the dual stochastic description, namely the B\"aumer equation, and analyze the crossover regimes of the propagator of this relativistic stochastic process. 

\section{Salpeter Hamiltonian}~\label{SEC:SalEq}

\begin{figure*}[t]
    \centering
    \includegraphics[trim=2cm 8cm 2cm 8cm, width=0.45\linewidth]{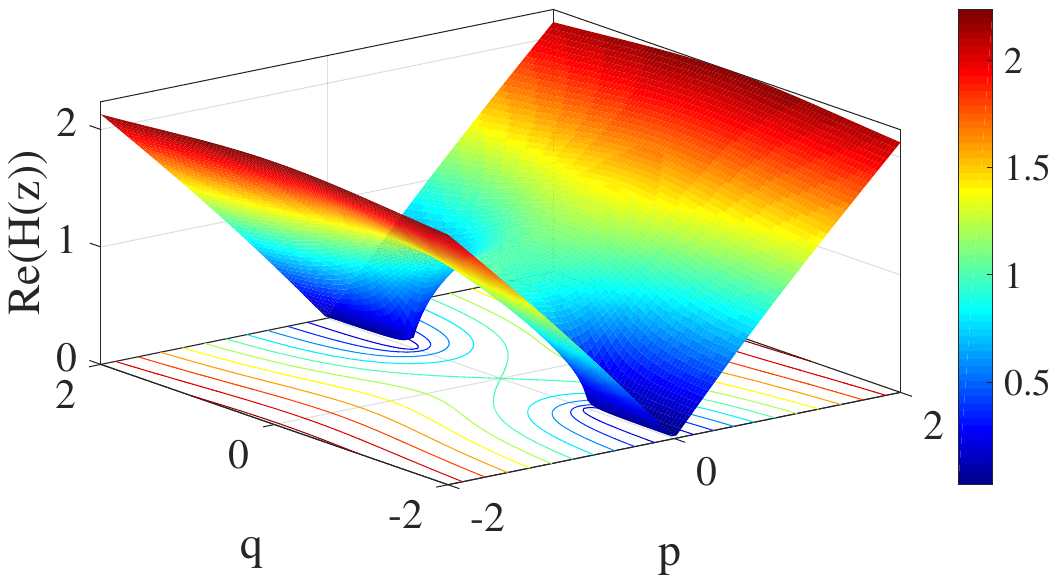}
    \includegraphics[trim=2cm 8cm 2cm 8cm, width=0.45\linewidth]{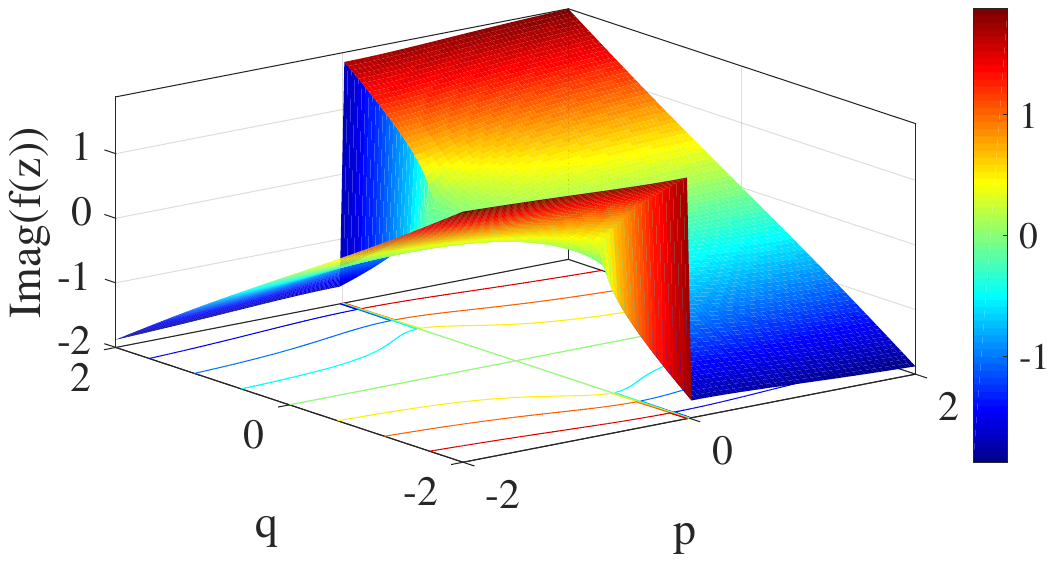}
    \caption{Real and imaginary part of the function $f(z)=\sqrt{1+z^2}$, where $z=p+iq$. The principal branch used in this calculation is obtained from the branch of the complex square root: $\sqrt{z}=\sqrt{\rho}(\cos\phi/2) + i*\sin(\phi/2))$, where  $z=\rho e^{i\phi}$ and $-\pi\le\phi\le\pi$. This branch is extended along the imaginary axis by imposing $\Im{f(iq)}>0 $ in the first quadrant, and $\Im{f(iq)}<0 $ in the fourth quadrant. This guarantee that the defined function is analytical in the first and fourth quadrant of the complex plane.}
    \label{fig:H_Complex}
\end{figure*}
Our derivation is based on the classical formulation of the RQM of a free particle, resulting to an integral representation of the Green (propagator) function, see Appendix \ref{app:GeneralFormalism}. The propagator function $G(x,t)$ is used to calculate the wave function $\psi(x,t)$  in terms of its initial value $\psi_0(x)$ as 

\begin{equation}
    \psi(x,t) =\int_{-\infty}^{\infty}{G(x-x_0,t)\psi_0(x_0)dx_0}.
    \label{eq:wavefunction_initial}
\end{equation}
The propagator depends on the Hamiltonian of the particle and is derived in Appendix \ref{app:GeneralFormalism} using textbook's quantum mechanics:
\begin{equation}
    G(x,t)=\frac{1}{2\pi}\int_{-\infty}^{\infty}{e^{-\frac{i}{\hbar} (H(p)t+ px)}dp}
    \label{eq:propagator_simple}
\end{equation}
Let us assume that the Hamiltonian is an even function of the momentum $H(-p)=H(p)$, There are two consequences from this property: First, the propagator will be an event function with respect to $x$,
i.e. $G(-x,t)=G(x,t)$. Second, the propagator can be calculated as the sum $G=G_{+}+G_{-}$ as
\begin{equation}
G_{\pm}(x,t)=\frac{1}{2\pi}\int_{0}^{\infty}{e^{-\frac{i}{\hbar} H(p)t \pm \frac{i}{\hbar} px}dp}
\label{eq:propagator_pm}
\end{equation}
Adding $G_{+}(x,t)$ and $G_{-}(x,t)$ we obtain an useful expression of the propagator $G(x,t)$ as
\begin{equation}
\begin{split}
    G(x,t) &=G_{+}(x,t)+G_{-}(x,t) \\ 
    &= \frac{1}{2\pi}\int_{0}^{\infty}{e^{-\frac{i}{\hbar} H(p)t}  (e^{\frac{i}{\hbar} px}+e^{-\frac{i}{\hbar} px})} dp \\
    &= \frac{1}{\pi}\int_{0}^{\infty}{e^{-\frac{i}{\hbar} H(p)t}  \cos(\frac{ px}{\hbar})} dp
\end{split}
\label{eq:G}
\end{equation}

Let us consider the Hamiltonian of a free, massive, and spin-less relativistic particle in one dimension, namely the Salpeter Hamiltonian 

\begin{equation}
    H_S(p) = \sqrt{m^2c^4 + c^2p^2}-mc^2
    \label{eq:Salpeter}
\end{equation}

Due to the square root, there is no standard differential equation for this Hamiltonian. In Eq.~(\ref{eq:Salpeter}) we subtract the rest energy $mc^2$ to recover the classical Schr\"odinger equation in the limit of small velocities $|p|<<mc$. This limit leads to
\begin{equation}
    H_S(p) \approx \frac{p^2}{2m}
    \label{eq:Salpeter_approx}
\end{equation}
Then the classical Schr\"odinger equation can be expressed in space coordinates as a partial differential equation
\begin{equation}
    i\partial_t\psi(x,t)=-\frac{\hbar^2}{2m}\partial_x^2\psi(x,t),
\end{equation}
A simple transformation of the Schr\"odinger equation in quantum mechanics can lead to the diffusion equation in a stochastic equation. This transformation can be done by either transforming time into imaginary $(t\to it)$ or transforming the Hamiltonian into imaginary $(H\to iH)$. Either transformation leads to the normal diffusion in the stochastic process 
\begin{equation}
    \partial_tP(x,t)=D\partial_x^2P(x,t),
\end{equation}
where $D=\hbar^2/(2m)$.  For the relativistic case, it is not possible to express the Sch\"odinger equation as a standard partial differential equation and the solution of the propagators is not well-known. Dirac was able to provide a relativistic equation by converting the wave equation into a four-dimensional function that inherits the spin of the particles. In this paper, we present an alternative by assuming that the wave function is a scalar so that it represents a relativistic quantum particle without spin. As shown in Appendix~\ref{app:GeneralFormalism}, the quantum bra-ket (bracket) formalism allows us to express the propagator as integral given by Eq.~(\ref{eq:propagator_simple}) without the need to rely on fractional differential equations. 
The analysis relies instead on performing the analytical continuation of the Salpeter Hamiltonian on the complex plane. First, the Hamiltonian is written as 
\begin{equation}
    H_S(p)=mc^2\left(f\left(\frac{p}{mc}\right)-1\right),
\end{equation}
where $f(p)=\sqrt{1+p^2}$. Then we define $f(z)$, where $z=p+iq$, as the analytical continuation of $f(p)$. Figure \ref{fig:H_Complex} shows the real and imaginary part of the principal analytical branch of the complex function $f(z)=\sqrt{1+z^2}$.  The real part of the function has a saddle point at $z=0$. The function is analytic in some regions containing the real axis $\Im(z)=0$. This is  where the real Hamiltonian function is defined. The complex function is analytical in the complex plane except for the two branch cuts $[-i\infty,-i)$ and $[i,i\infty)$ where the derivatives do not exist. We can slightly modify this function by using the function 
$f^*(p+iq)=\lim_{\epsilon\to 0}{f(p+q+i\epsilon)}$, where $\epsilon>0$. This will correct the values along the imaginary axis so that the new function becomes analytical in the first and fourth quadrants, including the imaginary axis.

For the sake of simplicity, for the rest of the paper we rescale the variables $x$, $t$, and $p$ in such a way that $\hbar = c =1$.  Thus Eq.~\ref{eq:Salpeter} becomes 
\begin{equation}
    H_S(p) = \sqrt{m^2 + p^2}-m
    \label{eq:Salpeter_a}
\end{equation}
Then substituting Eq. \ref{eq:Salpeter_a} into Eq \ref{eq:G}, the propagator we want to solve can be defined as
\begin{equation}
    G(x,t)=\frac{1}{2\pi}\int_{-\infty}^{\infty}{e^{-iH_S(p)t+ipx}dp}
    \label{eq:propagator_solve}
\end{equation}
We will also use the splitting the integral as $G=G_++G_-$ where
\begin{equation}
G_{\pm}(x,t)=\frac{1}{2\pi}\int_{0}^{\infty}{e^{- iH_S(p)t \pm i px}dp}
\label{eq:propagator_pm_a}
\end{equation}
The numerical solution to the above integrals is very unstable due to the heavy oscillatory behaviour of the integrands. We will show that more stable numerical solutions are possible by using complex variable analysis to obtain integral representation of the propagators. We will also provide analytical solutions for these integral representations.  In the next sections we present stable integral representations and analytical solutions in different cases. We will consider the analytical extension for the integrands and use complex variable analysis to handle the integrals. 

\subsection{Classical Limit: case \texorpdfstring{$m>0$}{m>0} and \texorpdfstring{$t\gg m$}{t m}}

We start with the integration of Eq.~\ref{eq:propagator_solve} by considering the case when $t$ is large. First, we perform an analytical continuation fo the integrand of Eq.~\ref{eq:propagator_solve} in the complex plane. We notice that the integrand of this equation is an analytical function in a region in the complex plane containing the real axis ($z=p+iq$ and $q=0$). We consider $t$ real and positive. When t is large, the major contribution of this integrand is when $H_S(p)$ is imaginary and minimal, which corresponds to the value $p=0$. Under these conditions, we can appeal to Debbie's method based on the {\it saddle point approximation}  presented in Section 31 of the Book of Dennery and Krzywicki \cite{dennery1996mathematics}. The method consists of replacing the Hamiltonian by its expansion up to the quadratic term around is minimal point $p=0$, i.e. $H_S(p) \approx p^2/(2m)$. Thus,
\begin{equation}
     G(x,t)\sim\frac{1}{2\pi}\int_{-\infty}^{\infty}{e^{-i\frac{p^2}{2m}t+ipx}dp} = \sqrt{\frac{m}{2\pi it}}e^{-\frac{mx^2}{2it}}
    \label{eq:propagator_classic}
\end{equation}
Here we use the following integration formula
\begin{equation}
    \int_{-\infty}^{\infty}{e^{iap-bp^2}dp}=\sqrt{\frac{\pi}{b}}e^{\frac{a^2}{4b}}\nonumber
\end{equation}
Note that this solution corresponds to the propagator of classical quantum mechanics. The evolution of a wave function can be obtained from its initial condition using Eq.~\ref{eq:convolution}. We can also recover the Green function of the diffusion equation by making the transformation $t\to it$ and $1/(2m)\to D$:
\begin{equation}
     G(x,t) \to \frac{1}{\sqrt{\pi Dt}}e^{-\frac{x^2}{Dt}}
    \label{eq:propagator_diffusion}
\end{equation}
The propagator accounts for the probability distribution function of the particle that was at $x=0$ at $t=0$. which corresponds to the classical diffusion process.

\subsection{Integral representation, Case \texorpdfstring{$m>0$}{m>0} and \texorpdfstring{$0<x<t$}{0<x<t}}

Here we derive an integral representation of the propagator whose integrand does not pose heavy oscillatory behaviour. The method is based on the Cauchy Theorem of complex variable analysis, that is used to convert the integral on the real axis into an integral over imaginary axis that removes heavy oscillations. The transformed integral are thus well-defined and easily integrable. The requirements to transform an integral over the real axis into the imaginary axis relies on two conditions: The complex continuation of the integrand is analytical in a quadrant of the complex plane; and the integrand vanishes as the complex argument is large enough in that quadrant. The details of the use of complex variable are included in two Theorems formulated and demonstrated in Appendix ~\ref{app:theorem}. Here  we will just show that more stable numerical solutions are possible by using these Theorems to obtain integral representation of the propagators. We will also provide analytical solutions for these integral representations.

First we provide an integral representation for the propagator of relativistic particles with arbitrary mass $m > 0$ inside the light cone ($0<x<t$). We replace Eq.~\ref{eq:Salpeter_a} into Eq.~\ref{eq:propagator_pm_a} to obtain
\begin{equation}
G_{\pm}(x,t)=\frac{1}{2\pi}\int_{0}^{\infty}{e^{\Phi_{\pm}(p)}dp},
\label{eq:propagator_pm_Phi}
\end{equation}
where
\begin{equation}
    \Phi_\pm(p) = -i\sqrt{m^2 + p^2}t -mt \pm ipx.
    \label{eq:Phi_SHamiltonian}
\end{equation}  
Now we consider that for $0<x<t$ and $-\pi/2<\theta<=0$ the condition 
\begin{equation}
\lim_{R\to\infty}{\Re{\Phi(R e^{i\theta})}}=\lim_{R\to\infty}{-R (-t \pm  x)\sin(\theta)}=-\infty.  
\label{eq:Phi_SHamiltonian_condition}
\end{equation}
This condition is satisfied in the fourth quadrant in the complex plane. Thus, applying Theorem 2 (See Appendix \ref{sec:Theorems}) to the analytical branch shown in Fig. \ref{fig:H_Complex} and substituting $p$ by $-iq$ we have
\begin{equation}
\begin{split}
&G_\pm(x,t)=\frac{e^{imt}}{2\pi}\int_{0}^{\infty}{e^{-i\sqrt{ m^2 + p^2}t \pm ipx}dp}\\
&=\frac{e^{imt}}{2\pi}\int_{0}^{\infty}{e^{-i\sqrt{ m^2 + (-iq)^2}t \pm i(-iq)x}d(-iq)}\\
&=\frac{e^{imt}}{2i\pi}\left(\int_{0}^{m}+\int_{m}^{\infty}\right){e^{-i\sqrt{ m^2 + (-iq)^2}t \pm qx}dq}\\
&=\frac{e^{imt}}{2i\pi}
\left(\int_{0}^{m}{e^{-i\sqrt{m^2-q^2}t \pm qx}dq}+\int_{m}^{\infty}{e^{-\sqrt{q^2-m^2}t \pm qx}dq}\right).
\end{split}
\label{eq:xlowert}
\end{equation}
The full propagator  $G(x,t)=G_+(x,t)+G_-(x,t)$ becomes
\begin{equation}
     G(x,t)=\frac{e^{imt}}{i\pi}\int_{0}^{\infty}{e^{-\sqrt{q^2-m^2}t} \cosh{(qx)}}dq
     \label{eq:Salpeter_inner}
\end{equation}
The integrand in this equation does not have singular oscillations and it easily integrated by standard numerical methods.

\subsection{Integral representation, Case \texorpdfstring{$m>0$}{m>0} and \texorpdfstring{$x>t>0$}{x>t>0}}
Now we obtain the expression for the propagator outside of hte light cone ($x>t>0$). Under these conditions $\Phi_-$ in Eq.~\ref{eq:Phi_SHamiltonian} satisfies the condition Eq.~\ref{eq:Phi_SHamiltonian_condition} in the four quadrant where $\sin{\theta}<0$. For $\Phi_+$ we note that $-t+x>0$.  For this case the integral is calculated in the first quadrant where $\sin{\theta}>0$. Therefore we apply only the loop integral for the first quadrant for $\Phi_+$ and for the fourth quadrant for $\Phi_-$ to Eq.~\ref{eq:propagator_pm_a}.
Applying Theorem 2 on $G_-$ for the first quadrant and using the change of variable $p\to iq$
\begin{equation}
  \begin{split} 
  &G_{-}(x,t)=\frac{e^{imt}}{2\pi}\int_{0}^{\infty}{e^{-i\sqrt{ m^2 + p^2}t - ipx}dp}\\
            &=\frac{e^{imt}}{2\pi}\int_{0}^{\infty}{e^{-i\sqrt{ m^2 + (-iq)^2}t - i(-iq)x}d(-iq)}\\
            &=\frac{e^{imt}}{2i\pi}\left(\int_{0}^{m}{e^{-i\sqrt{m^2-q^2}t -qx}dq}
                                  +\int_{m}^{\infty}{e^{-\sqrt{q^2-m^2}t -qx}dq}\right).
\end{split}  
\label{eq:xlargert-}
\end{equation}
And applying Theorem 1 on $G_+$ for the fourth quadrant and using the change of variable $p\to -iq$
\begin{equation}
\begin{split}
&G_{+}(x,t)=\frac{e^{imt}}{2\pi}\int_{0}^{\infty}{e^{-i\sqrt{ m^2 + p^2}t + ipx}dp}\\
&=\frac{e^{imt}}{2\pi}\int_{0}^{\infty}{e^{-i\sqrt{ m^2 + (iq)^2}t + i(iq)x}d(iq)}\\
&=\frac{-e^{imt}}{2i\pi}\left(\int_{0}^{m}{e^{-i\sqrt{m^2-q^2}t - qx}dq}
                        +\int_{m}^{\infty}{e^{\sqrt{q^2-m^2}t - qx}dq}\right). 
\end{split}
\label{eq:xlargert+}
\end{equation}
Adding both parts and using the identity $1/i = -i$, the full propagator becomes
\begin{equation}
     G(x,t)=\frac{ie^{imt}}{\pi}\int_{m}^{\infty}{\sinh{(\sqrt{q^2-m^2}t)} e^{-qx}dq}
     \label{eq:Salpeter_outer}
\end{equation}

Note that this solution corresponds to non-zero solution for the propagator outside of this light cone, suggesting faster-than-light propagation. While the relativistic quantum mechanic formulation allows for the possibility of particles having speeds higher than light, the green function decays rapidly outside the light cone. This is associated with certain challenges and inconsistencies, such as causality violations and the lack of a clear physical interpretation. Quantum field theories, by treating particles as excitations of a vacuum state, --i.~e. quantum fields-- resolve this problem, where the Green function is expressed explicitly in terms of the expectation value of quantum fields. More explicitly, one considers the commutation (anticommutation) relations for bosonic and fermionic creation and destruction operators. As a result the retarded and advanced Green functions are found to be zero (non-zero) outside (inside) the light cone, the expression of which is different for fermions and bosons~\cite{peskin2018introduction}. 

\subsection{Massless particles: Case $m=0$}
 
When the particle has no mass, the propagator for $0<x<t$ is calculated by taking $m=0$ in the last line of  Eq. ~\ref{eq:xlowert}
\begin{equation}
  G_\pm(x,t)=\frac{1}{2i\pi}\int_{0}^{\infty}{e^{-qt\pm qx}dq}=\frac{1}{2i\pi}\frac{-1}{-t\pm x}.
\end{equation}
Adding both parts, the full propagator of mass-less particles for $x<t$ becomes
\begin{equation}
    G(x,t) = \frac{1}{\pi}\frac{it}{x^2-t^2}
    \label{eq:massless-}
\end{equation}
For the case $0<t<x$ the propagator is obtained taking $m=0$ in Eq. ~\ref{eq:xlargert-} and ~\ref{eq:xlargert+}
\begin{equation}
\begin{split}
  &G_{-}(x,t)=\frac{1}{2i\pi}\int_{0}^{\infty}{e^{-qt -qx}dq}=\frac{1}{2i\pi}\frac{1}{x+t}\\
  &G_{+}(x,t)=\frac{i}{2\pi}\int_{0}^{\infty}{e^{qt - qx}dq}=\frac{i}{2\pi}\frac{1}{x-t}\\
\end{split}
\end{equation}
Adding both parts, the full propagator of mass-less particles for $x>t$ becomes
\begin{equation}
    G(x,t) = \frac{1}{\pi}\frac{it}{x^2-t^2}
    \label{eq:masslesspropagator}
\end{equation}
This result is the same that Eq. \ref{eq:massless-} for case $x<t$. We notice a singularity at the light cone $x=t$ and, again, non-zero values at $x>t>0$ indicating faster-than-light-propagation. As in the case of massive particles, a quantum field theory formulation is required to avoid the apparent violation of the causality principle in the classical interpretation of quantum mechanics.\\

\begin{figure*}[t]
    \includegraphics[scale=0.58,clip]{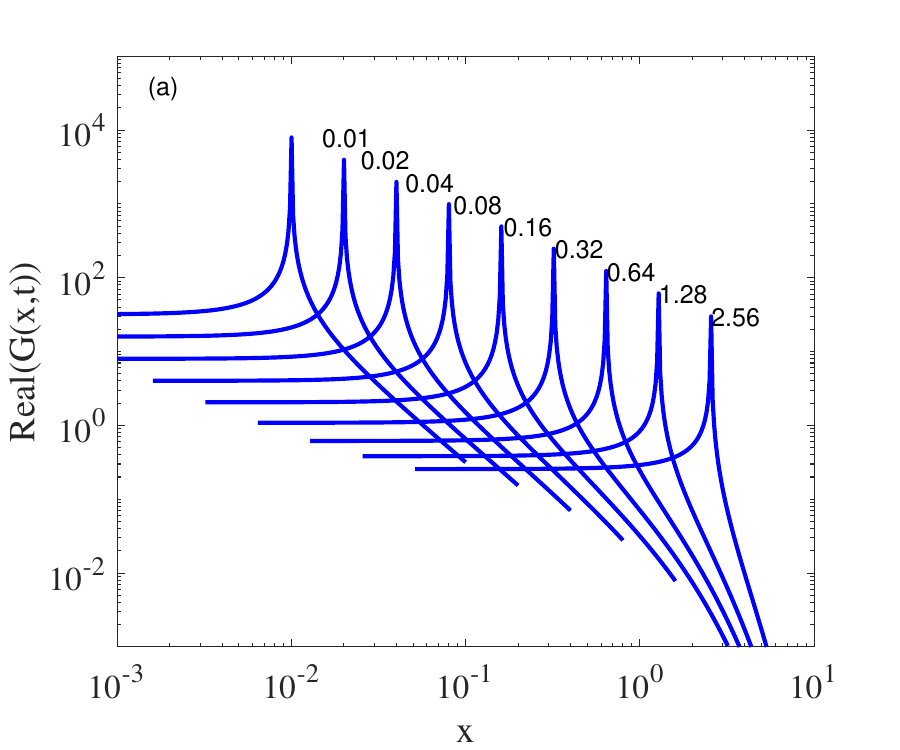}
    \includegraphics[scale=0.58,clip]{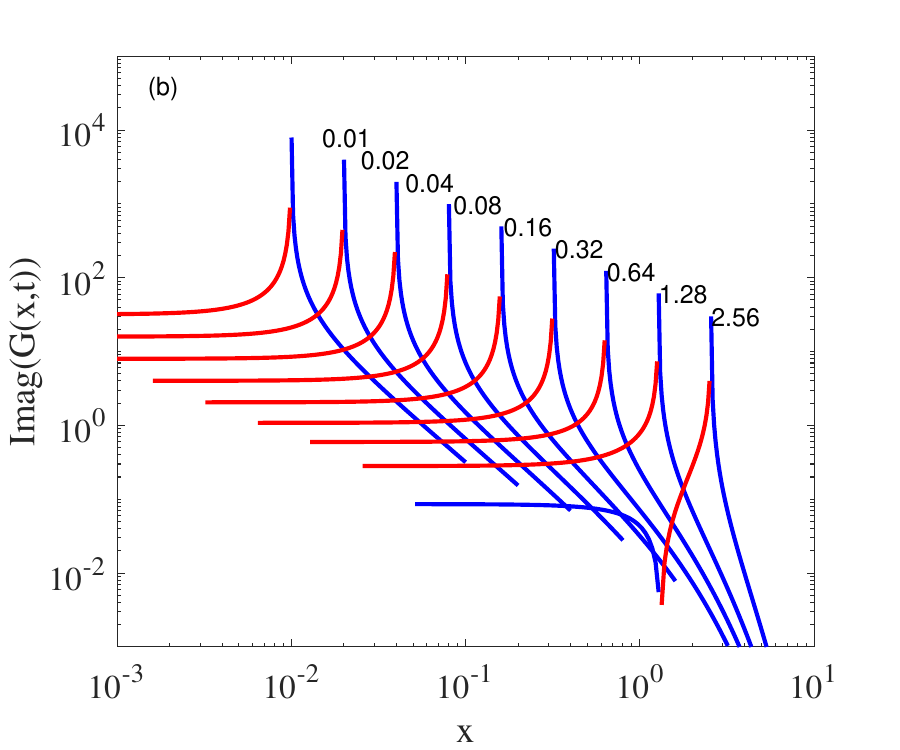}\\
    \includegraphics[scale=0.58,clip]{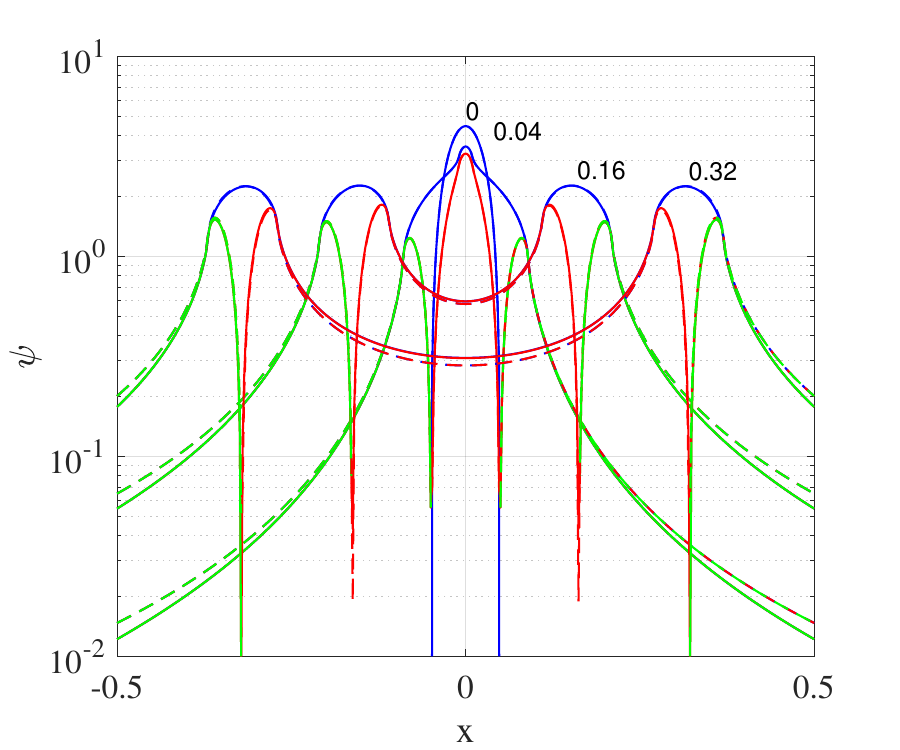}
    \includegraphics[scale=0.58,clip]{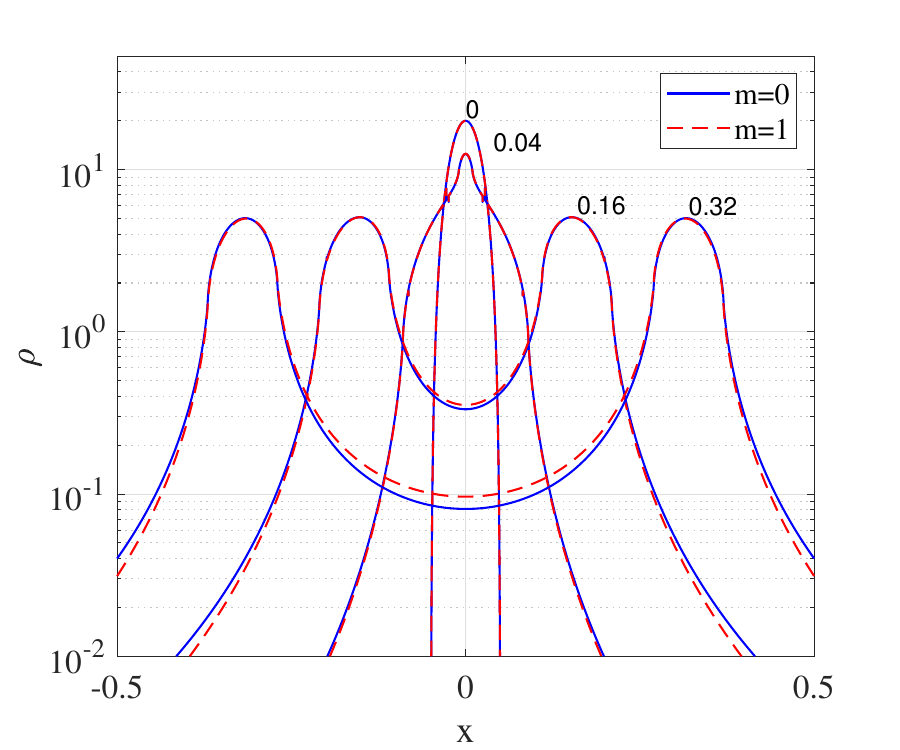}
    \caption{(a) Comparison between the numerical and analytical solution of the Salpeter propagator $G(x,t)$ concerning $x$, across various $t$ values. (a) shows the real part and (b) the positive (blue) and negative (red) values of the imaginary parts in log-log scale. The findings from both methods exhibit strong consistency and alignment. Notably, the propagator displays a singularity at the light cone, represented by $x=t$. (b) The phase change compared between the numerical and analytical solutions, depicted in terms of $\frac{1}{2}-\frac{\theta(x,t)}{\pi}$ with respect to $x$. For all $t$ values, the phase consistently converges to $-\pi/2$ in both the numerical and analytical solutions. (c) and (d) are the Wave function calculated using Eq. \ref{eq:avantiwave}. The real part is plotted in blue, and the imaginary part is plotted in green. (c) For a massless particle $m=0$. (d) For a particle $m=1$.}
    \label{fig:Salpeter_G_compare}
\end{figure*}

\subsection{Analytical Solution}
Within this subsection, we establish the analytical solution for the Salpeter propagator. Substitute the Salpeter Hamiltonian (Eq. \ref{eq:Salpeter_a}) into the propagator (Eq.~\ref{eq:G}) we have
\begin{equation}
\begin{split}
     %G(x,t) &= \frac{1}{\pi}\int_{0}^{\infty}{e^{-i \hbar (\sqrt{m^2c^4 + c^2p^2}-mc^2)t}  cos(\hbar px)} dp \\
     G(x,t) &= \frac{1}{\pi}\int_{0}^{\infty}{e^{-i (\sqrt{m^2 +p^2}-m)t}  \text{cos}(px)} dp \\
     %&= \frac{1}{\pi} e^{i\hbar mc^2t} \int_{0}^{\infty}{e^{-i\hbar \sqrt{m^2c^4 + c^2p^2}t}  cos(\hbar px)} dp \\
     &= \frac{1}{\pi} e^{i mt} \int_{0}^{\infty}{e^{-i \sqrt{m^2+ p^2}t}  \text{cos}(px)} dp \\
\end{split}
\end{equation}
Using the identity \cite{kowalski2011salpeter}
\begin{equation}
\begin{split}
        \int_{0}^{\infty} &{dx \,\, \text{exp}(-\alpha \sqrt{x^2+\beta^2}) \text{cos}\gamma x } \\
    &= \frac{\alpha \beta}{\sqrt{\alpha^2+\gamma^2}}K_1(\beta\sqrt{\alpha^2+\gamma^2}), \text{Re}[\alpha]>0, \text{Re}[\beta]>0
\end{split}
\label{eq:integral}
\end{equation}
We can rewrite the propagator as
\begin{equation}
\begin{split}
     %G(x,t) &=\frac{1}{\pi} e^{i\hbar mc^2t} \int_{0}^{\infty}{e^{-\alpha \sqrt{\beta^2+p^2}}  cos(\gamma p)} dp \\
     G(x,t) &=\frac{1}{\pi} e^{i mt} \int_{0}^{\infty}{e^{-\alpha \sqrt{\beta^2+p^2}}  \text{cos}(\gamma p)} dp \\
\end{split}
\end{equation}
where $\alpha=it$, $\beta=m$, and $\gamma=x$.
Then according to the right-hand side of Eq.~\ref{eq:integral}, we have
\begin{equation}
\begin{split}
     %G(x,t) &=\frac{imt e^{i\hbar mc^2t}}{\pi \sqrt{x^2-\frac{t^2}{c^2}}} K_1(mc \hbar \sqrt{x^2-\frac{t^2}{c^2}}) dp \\
     G(x,t) &=\frac{imt e^{i mt}}{\pi \sqrt{x^2-t^2}} K_1(m\sqrt{x^2-t^2}),\ x^2>t^2
\label{eq:sal_analy+}
\end{split}
\end{equation}
where $K_1$ is the modified Bessel function of the second kind.
For $x^2<t^2$, we can utilize the property of the modified Bessel functions, given as:
\begin{equation}
    K_\alpha (x)=\frac{\pi}{2} i^{\alpha +1} H_\alpha ^{(1)} (ix), \,\,  -\pi<\text{arg} \, x \le \frac{\pi}{2}
    \label{eq:B_to_H}
\end{equation}
where $H_\alpha ^{(1)}$ is the Hankel function of the first kind. 
Replacing $x^2-t^2$ by $-(t^2-x^2)$ in Eq.~\ref{eq:sal_analy+} we have:
\begin{equation}
\begin{split}
    G(x,t) &=\frac{imt e^{i mt}}{\pi \sqrt{-(t^2-x^2)}} \frac{\pi}{2} i^2 H_1^{(1)}(im \sqrt{-(t^2-x^2)}) \\
    &=-\frac{mt e^{i mt}}{2\sqrt{t^2-x^2}} H_1^{(1)}(-m \sqrt{t^2-x^2}),\ x^2<t^2
    \label{eq:sal_analy-}
\end{split}
\end{equation}
In summary, the analytical solution of the Salpeter propagator is
\begin{equation}
G(x,t)  =
\begin{cases}
    \frac{imt e^{imt}}{\pi\sqrt{x^2-t^2}}K_1\left(m\sqrt{x^2-t^2}\right), & x^2>t^2\\
    -\frac{mt e^{imt}}{2\sqrt{t^2-x^2}}H_1^{(1)}\left(-m\sqrt{t^2-x^2}\right), & x^2<t^2,
\end{cases}
\label{eq:sal_ana_summary}
\end{equation}
where $H_1^{(1)}$ is the Hankel function of first kind, and $K_1$ is the modified Bessel function of the second kind.

In Fig.~\ref{fig:Salpeter_G_compare}, we present the numerical and analytical solutions in terms of the real and imaginary part of the propagator and the phase at varying $t$. Both solutions are indistinguishable. We also compare the cases $m=1$ and $m=0$. The numerical solution inside the light cone $x<t$ is obtained by numerical integration of Eq~\ref{eq:Salpeter_inner}.
The numerical solution outside the cone $x>t$ is the numerical integration of Eq.~\ref{eq:Salpeter_outer}.  We note that the hyperbolic functions of these integrals are sensitive to large arguments, but this issue can be solved by expressing them in terms of exponential functions.
The analytical solution is calculated by substituting the corresponding values of $x$ and $t$ into Equation \ref{eq:sal_ana_summary}. In this context, we choose $t=2^n/100, n=0,1,2...8$, and $x \in [0,5t]$ as the calculation range.  The solutions for $m=1$ and $m=0$ are asymptotically similar when $d=|x-t|\to 0$. Here $d$ is the distance to the light cone $|x|=t$. For each time $t$, the propagators present a singularity at $d=0$. This singularity proves to be of the order $1/d$ so that it produces a sudden change in the values of the propagator and a discontinuity of the argument of this function. As we will see in the next section, this singularity should be carefully treated once we calculate the wave function of the Salpeter equation.

\subsection{Salpeter wave function}

Using the Eq.~\ref{eq:wavefunction_initial}, one finds the wave function in later times, given its form in the initial time. After the change of variable $x'=x-x_0$ this equation can be written as
\begin{equation}
    \psi(t,x) =\int_{-\infty}^{\infty}{G(x',t)\psi_0(x-x')dx'}.
\end{equation}
Strictly speaking, this integral does not exist because the propagator has two first order singularities at $x =\pm t$. To obtain a meaningful solution for the wave function, we need to circumvent the singularities by choosing an appropriate integration path in the complex plane. A meaningful solution is obtained if the integration is performed on the deformed path $C$ shown in Fig~\ref{fig:avanti}. We redefine the wave function as the path integral 
\begin{equation}
    \psi^*(t,x) =\int_C{G(z,t)\psi_0(x-z)dz}.
\end{equation}
This integral can be performed by applying the theorem of residues of the complex variable \cite{dennery1996mathematics} over the closed path shown in Fig~\ref{fig:avanti}. After calculating the residues inside the closed path, the integration becomes:
\begin{equation}
\begin{split}
     \psi^*(t,x) &= PV \int_{-\infty}^{\infty}{G(x',t)\psi_0(x-x')dx'}\\
    &+\frac{1}{2}\psi_0(x-t)+\frac{1}{2}\psi_0(x+t).
\end{split}
\label{eq:avantiwave}
\end{equation}
where $PV$ refers to the principal value of the integral \cite{dennery1996mathematics}. To see the proof of Eq.~\ref{eq:avantiwave}, and other details, refer to Appendix \ref{app:theoremofResidues}.
For the initial condition $\psi_0(x)$ we choose a wave function with compact support in the interval $[-\delta/2,\delta/2]$. The simplest case is to assume that the wave function is initially constant and real on the interval, by normalization we found that the initial wave function is constant in this interval and zero elsewhere $\psi_0(x) = \delta^{-1/2}(|x|<\delta/2)$. Using this initial condition turns out to be problematic since the wave function is non-zero at $x \pm \delta$, and it leads to four singularities in the principal value in Eq.~\ref{eq:avantiwave} at $|x-t|\pm \delta$. To remove these singularities, we need to assume that the initial wave function is zero at $x \pm \delta$. In other words, we need to guarantee that the initial wave function is continuous. Our next best option for the initial condition is an initially harmonic wave function:
\begin{equation}
   \psi_0(x) = \sqrt{\frac{2}{\delta}}\cos{\frac{\pi x}{\delta}}(|x|<\delta/2) 
   \label{eq:wave_initial}
\end{equation}
For $m>0$, We calculate the wave function based on the numerical integration of the propagator, and the results are plotted in Figure \ref{fig:Salpeter_G_compare}c. In spite of the singularity of the propagator at $x=\pm t$, the wave functions are shown to be continuous and differentiable everywhere. The propagation of the wave function is somewhat similar to the one in the wave equation, where one initial pulse is split into two equal pulses moving left and right. The main difference here is that the pulses are complex and include not only propagation but also dissipation. The probability density is calculated as $\rho = \phi^* \phi$ and shown in Figure
~\ref{fig:Salpeter_G_compare}d. The continuity equation is verified by calculating the limit of the area under 
$\rho$ for different values of the value $\epsilon$ used in the calculation of the principal value of the wave functions.  It is verified that for all values of time used in Figure \ref{fig:Salpeter_G_compare}, the limit of the integral is very close to one with an error below $0.004\%$. This is consistent with the continuity equation for the Salpeter equation, already demonstrated by Kowalski and Rembielinski~\cite{kowalski2011salpeter}

\begin{figure}[t]
    \includegraphics[scale=0.4,trim={70 140 20 40},clip]{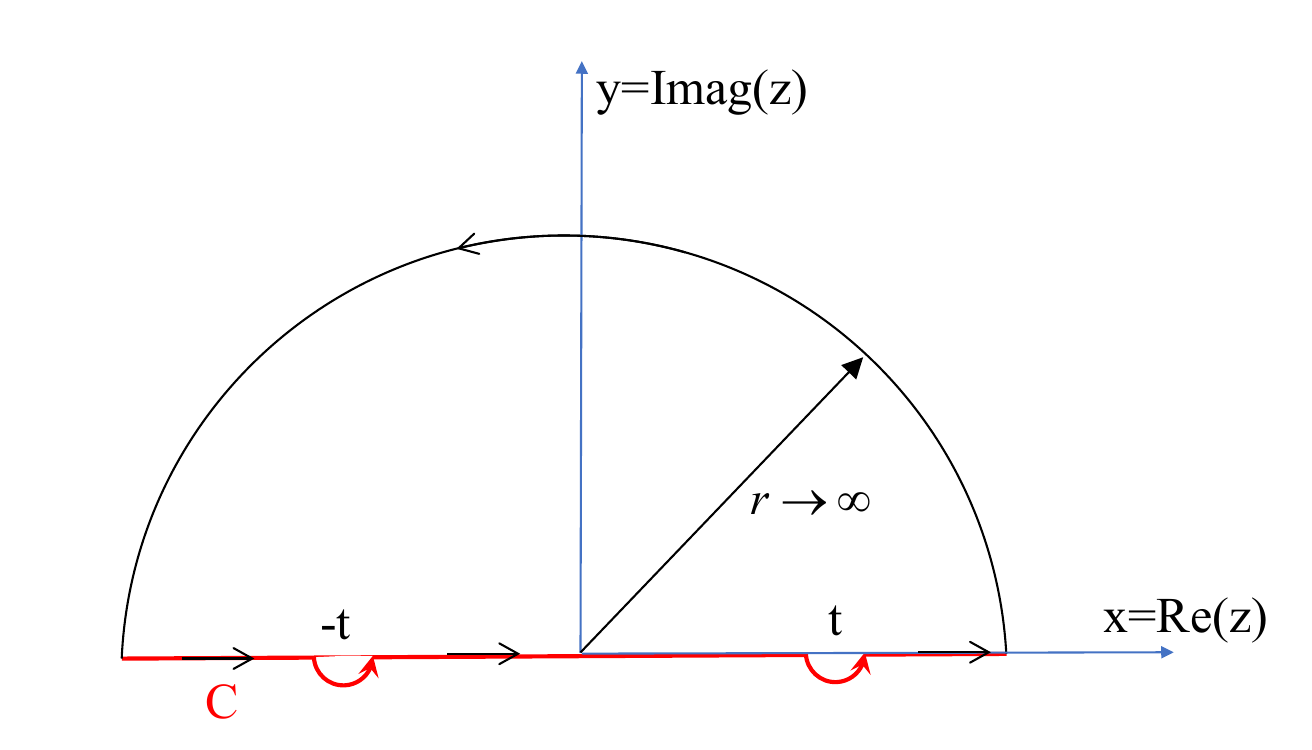}
    \caption{Contour in the complex plane used to calculate the Salpeter wave function.}
    \label{fig:avanti}
\end{figure}

\section{B\"aumer Hamiltonian}~\label{SEC:RelStoch}
Now we consider the Relativistic Diffusion Equation (RDE), which is obtained by a Wick rotation, or equivalently, multiplying the Hamiltonian by the imaginary unit $i$. The corresponding Hamiltonian has been investigated by B\"aumer et al.~\cite{baeumer2010stochastic} and will be referred here as the B\"aumer Hamiltonian 
\begin{equation}
    H_B(p)=i\sqrt{m^2 + p^2}-im.
    \label{eq:Baeumer_H}
\end{equation}
Replacing this Hamiltonian into the Schr\"odinger equation \ref{eq:schrodinger}, and using Eq.~\ref{eq:momentum} the RDE is obtained as the following fractional differential equation:
\begin{equation}
    2D[ \sqrt{1+\frac{\partial^2}{\partial x^2}} -1]P = \frac{\partial P}{\partial t},
\end{equation}
where $D\equiv \frac{m}{2}$ and $P$ is the probability density in the stochastic space, equivalent to the wave function in the quantum space. The quantum mechanics' mathematical formalism allows a simple solution of this fractional differential equation by solving the propagator that in this case is the Green function of the RDE.

The propagator in Eq.~\ref{eq:propagator} for the Hamiltonian given by Eq.~\ref{eq:Baeumer_H} becomes
\begin{equation}
    G(x,t)=\frac{1}{2\pi}\int_{-\infty}^{\infty}{e^{-\sqrt{m^2 + p^2}t+mt+ipx}dp}
    \label{eq:propagator_RDE}
\end{equation}
whose components are
\begin{equation}
    G_\pm(x,t)=\frac{1}{2\pi}\int_{0}^{\infty}{e^{-\sqrt{m^2 + p^2}t+mt\pm ipx}dp}
    \label{eq:propagator_pm_RDE}
\end{equation}
We are considering  the different cases to discuss the solution below:

\subsection{Classical Limit: Case \texorpdfstring{$m>0$}{m>0} and \texorpdfstring{$t\gg m$}{t m}}
The classical limit is the same as the asymptotic solution when $t\to\infty$. In such case the propagator \ref{eq:propagator_RDE} becomes
\begin{equation}
    G(x,t)=\frac{1}{2\pi}\int_{-\infty}^{\infty}{e^{-(p^2/2m)t+ipx}dp}
    \label{eq:propagator_G_RDE}
\end{equation}
The solution corresponds to a Gaussian distribution 
\begin{equation}
     G(x,t)=\sqrt{\frac{m}{{2\pi t}}}e^{-\frac{mx^2}{2t}}
    \label{eq:propagator_diffusion_gauss}
\end{equation}

 We can also recover the Green function of the diffusion equation by making the transformation $t\to it$ in the classical Schr\"odinger equation and taking $1/(2m)\to D$:
\begin{equation}
     G(x,t) \to \frac{1}{\sqrt{\pi Dt}}e^{-\frac{x^2}{Dt}}
\end{equation}

\subsection{Integral representation for \texorpdfstring{$m>0$}{m>0} and \texorpdfstring{$0\le x<\infty$}{0<x<inf} (Inner solution)}

For the general case, the components in Eq.\ref{eq:propagator_pm_RDE} can be summed up to obtain; 
\begin{equation}
     G(x,t)=\frac{e^{mt}}{\pi}\int_{0}^{\infty}{e^{-\sqrt{m^2 + p^2}t}\cos{(px)}dp}
     \label{eq:Baeumer_inner}
\end{equation}

We note that the integrand of has strongly oscillatory behavior as $x\to\infty$. This integral is therefore the {\it inner solution} of the propagator that can be applied for moderate values of $x$ but will not effectively converge at the tails. An outer solution is sought in the next section using complex variable analysis. 

\subsection{Integral representation for \texorpdfstring{$m>0$}{m>0} and \texorpdfstring{$x>0$}{x>0} (Outer solution)}

We write the two terms of the propagator as Eq.~\ref{eq:propagator_pm_Phi}, 
\begin{equation}
G_{\pm}(x,t)=\frac{1}{2\pi}\int_{0}^{\infty}{e^{\Phi_{\pm}(p)}dp},
\nonumber
\end{equation}
 and the asymptotic behaviour of $\Phi_\pm(z)=-iH_B(z)t\pm xz$ for $z=Re^i\theta$ as $R\to\infty$ follows
\begin{equation}
\lim_{R\to\infty}{\Re{\Phi(R e^{i\theta})}}=\lim_{R\to\infty}{-R (t\cos{\theta}\pm x\sin{\theta}})
\label{eq:condition_RDE}
\end{equation}
The condition given by Eq.~\ref{eq:condition_RDE} is satisfied in the first quadrant of the complex plane for $G_+$ and in the fourth quadrant for $G_-$. The integration of $G_+$ over the first quadrant gives
\begin{equation}
\begin{split}
    & G_+(x,t)=\frac{1}{2\pi}\int_{0}^{\infty}{e^{-\sqrt{m^2 + p^2}t +mt+ ipx}dp}\\
     &=\frac{e^{mt}}{2\pi}\int_{0}^{\infty}{e^{-\sqrt{m^2 + (iq)^2}t + i(iq)x}d(iq)}\\
      &=\frac{ie^{mt}}{2\pi}\left(\int_{0}^{m}{e^{-\sqrt{m^2  -q^2}t-qx}dq}+\int_{m}^{\infty}{e^{-i\sqrt{q^2-m^2}t-qx}dq}\right)\\\nonumber
\end{split}
\end{equation}
The integration of $G_-$ over the four quadrants results in
\begin{equation}
\begin{split}
    &G_-(x,t)=\frac{1}{2\pi}\int_{0}^{\infty}{e^{-\sqrt{m^2 + p^2}t+mt - ipx}dp}\\
    &=\frac{e^{mt}}{2\pi}\int_{0}^{\infty}{e^{-\sqrt{m^2 + (-iq)^2}t - i(-iq)x}d(-iq)}\\
    &=\frac{-ie^{mt}}{2\pi}(\int_{0}^{m}{e^{-\sqrt{m^2 -q^2} -qx}dq}
    +\int_{m}^{\infty}{e^{i\sqrt{q^2 -m^2} -qx}dq})
    \nonumber
\end{split}
\label{eq:propagatorRDE_n}
\end{equation}
Adding both parts and the full propagator becomes
\begin{equation}
     G(x,t)=\frac{-e^{mt}}{\pi}\int_{m}^{\infty}{\sin{(\sqrt{q^2-m^2}t)} e^{-qx}dq}.
     \label{eq:Baeumer_outer}
\end{equation}
Interestingly, Using the identity $sinh(iz) = isin(z)$ and replacing $t$ by $it$ we can recover the outer-cone solution of the Salpeter equation, Eq.~\ref{eq:Salpeter_outer}.

\subsection{Case \texorpdfstring{$m=0$}{m=0}}
For massless stochastic particles, the components of the propagator in Eq.~\ref{eq:propagator_pm_RDE} becomes 
\begin{equation}
    G_\pm(x,t)=\frac{1}{2\pi}\int_{0}^{\infty}{e^{-pt\pm ipx}dp}
    \label{eq:propagator_m0_RDE}
\end{equation}
Both integrals can be analytically integrated. Summing up, the solution is
\begin{equation}
    G_{m=0}(x,t)=\frac{1}{\pi}\frac{t}{t^2+x^2}
    \label{eq:propagator_m0_A}
\end{equation}
This is the Cauchy distribution with scale parameter given by $t$.
\subsection{Case \texorpdfstring{$x \gg t$}{x  t}}
The behavior of the tail of the B\"aumer distribution can be obtained by using $x>>t$ in Eq.\ref{eq:Baeumer_outer}. 
\begin{equation}
     G(x,t)=\frac{-e^{m(t-x)}}{\pi}
     \label{eq:Baeumer_tails}
\end{equation}
Interestingly, the B\"aumer distribution exhibits exponential-tailed behavior that is not the case for either the Cauchy distribution or the Gaussian distribution.

\begin{figure}[b]
    \centering
    \includegraphics[trim=1.8cm 8cm 1.8cm 8cm, width=\linewidth]{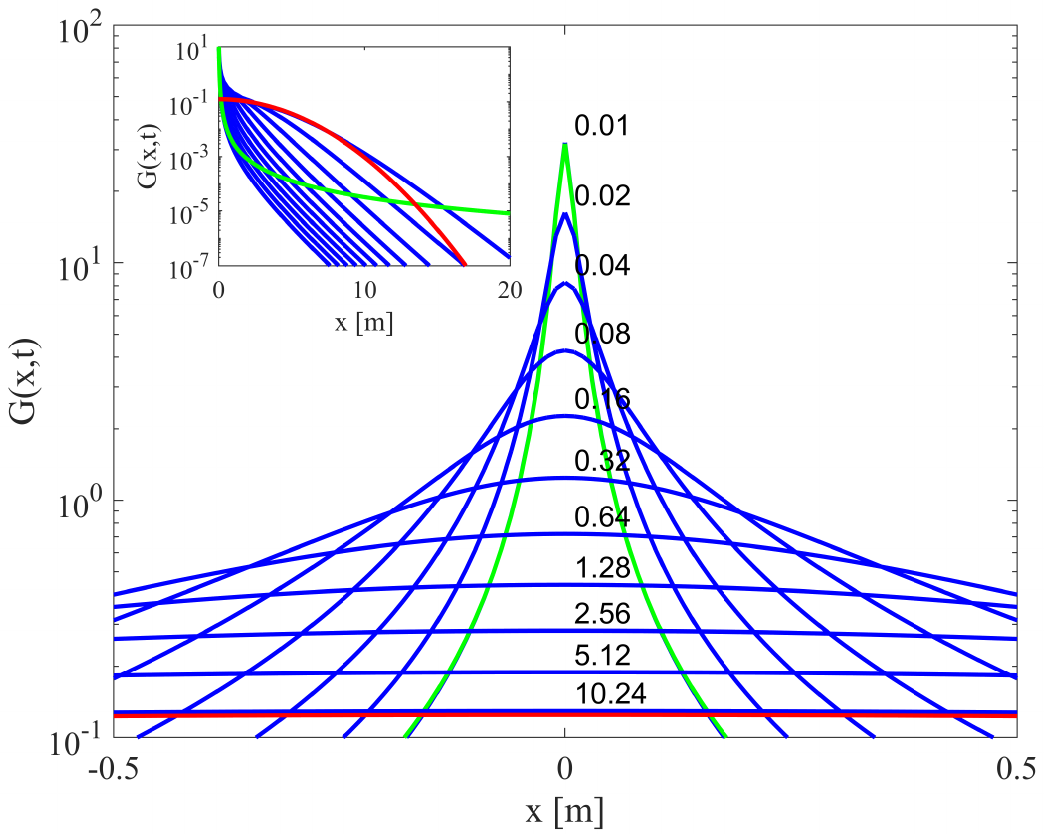}
            \caption{Solution of the B\"aumer propagation that describes the relativistic diffusion propagator. $G(x,t)$  is plotted in blueagainst $x$ for different values of $t$. The propagator is compared to the Cauchy distribution (green) and Gaussian distribution (red). The inset shows the exponential law of the tails.} 
            \label{fig:Baeumer_G}
\end{figure}

\begin{figure}[b]
    \centering
    \includegraphics[trim=1.8cm 8cm 1.8cm 6cm, width=0.9\linewidth]{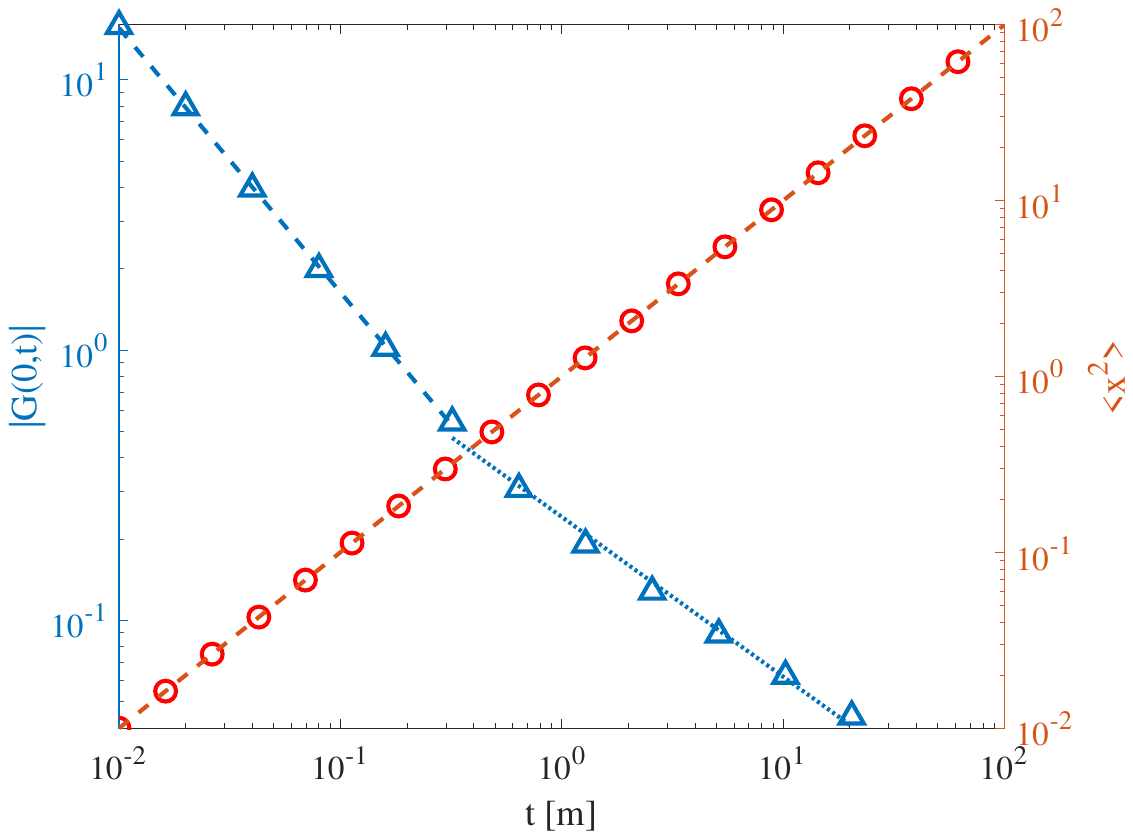}
    \caption{Evolution of the second moment of the B\"aumer propagator $<x^2>$ against $t$. The solid line corresponds to power law $y=t$}
    \label{fig:Baeumer_Diffusion}
\end{figure}

\subsection{Analytical solution}
The analytical solution of the B\"aumer propagator is obtained by substituting the imaginary Hamiltonian into the Salpeter propagator
\begin{equation}
\begin{split}
     G(x,t) &= \frac{1}{\pi}\int_{0}^{\infty}{e^{-i^2 (\sqrt{m^2 +p^2}-m)t}  \text{cos}(px)} dp \\
     &= \frac{1}{\pi} e^{-mt} \int_{0}^{\infty}{e^{\sqrt{m^2+ p^2}t}  \text{cos}(px)} dp \\
\end{split}
\end{equation}

Using the identity in Eq. \ref{eq:integral}, we have $\alpha=-t, \beta=t$ and $\gamma=x$, and the solution to the integral can be written as:
\begin{equation}
    G(x,t) = -\frac{mt e^{-mt}}{\pi \sqrt{x^2+t^2}} K_1 (m \sqrt{x^2+t^2})
\end{equation}

Note that the solution of the propagator defines a diffusion process with a transition from Cauchy to diffusion, shown in 
Fig.~\ref{fig:Baeumer_G}. The propagator (i.e., The Green function) converges asymptotically to a Cauchy distribution as $t\to 0$, a Gaussian distribution as $t\to\infty$, and an exponential function when $x\to\infty$. 

The B\"aumer propagator defines a diffusion process that is more complex than Gaussian diffusion. In many cases, the diffusivity is characterized by a power law on the evolution of the peak $G(0,t)\sim t^{-\alpha}$, or the power law of the evolution of the second moment $\langle x^2G(x,t)\rangle_x\sim t^{-\alpha}$. The diffusion process corresponds to normal diffusion when $\alpha=2$, super-diffusion when $\alpha>2$, and sub-diffusion when  $\alpha<2$.
Fig.~\ref{fig:Baeumer_Diffusion} shows the time evolution of both the peak and second moment of the propagator. The evolution of the peak shows a crossover between subdiffusion ($\alpha =1$) for small times, and normal diffusion ($\alpha=2$) for large times. A different diffusion process results from the second moment that exhibits normal diffusion at all times. This complexity of the diffusion process with crossover regions over space-time is also observed in the temporal evolution of the PDFs of the predicted price return \cite{alonso2019q}. However, the values of the exponents are different from the B\"aumuer diffusion process. In particular, the diffusion process in stock markets is not only fractional but also nonlinear. Thus, the use of quantum mechanics in stock market prediction is still severely limited due to the current limitations on the mathematical description of the nonlinearity of quantum mechanics. On the other hand, the link between stochastic equations and quantum mechanics by the Wick rotation provides interesting possibilities, like the option of solving complex non-linear stochastic problems, such as the Porous Media Fractional Equation \cite{tang2024variable,gharari2021space} and then using the Wick rotation to obtain the corresponding non-linear quantum problems.

\section{Conclusions}
We presented the numerical, analytical, and asymptotic solutions of the Salpeter propagator for relativistic quantum mechanics of a free particle with no spin, and its stochastic counterpart, the  B\"aumer propagator. The solution of the B\"aumer propagator converges asymptotically to Cauchy distribution for small times, and to Gaussian distribution with large times. The Salpeter propagation is recovered by transforming time by imaginary time. The solution of this quantum mechanical propagator shows important properties, such as non-zero values outside of the light cone, propagation with diffusion around the light cone, and singularity of the light cone. Special treatment is required to handle this singularity in the calculation of the wave function using the concept of principal value of integrals. A compact-supported initial wave function with no discontinuities at the edges allows us to remove the singularity, leading to well-posed wave functions. The conservation of probability is checked by testing that the integral of the probability density remains one for all times.

This paper has focused on the use of complex variable analysis for the derivation of solutions and the careful testing of these solutions using numerical analysis. The interpretation of these solutions in the context of relativistic quantum mechanics, quantum field theories, and relativistic stochastic dynamics is a matter of future publication.

\providecommand*{\mcitethebibliography}{\thebibliography}
\csname @ifundefined\endcsname{endmcitethebibliography}
{\let\endmcitethebibliography\endthebibliography}{}

%end of bbl file copied-pasted

%\end{bibliography}

\appendix
\section{Quantum Propagator}
\label{app:GeneralFormalism}

We can write the Schr\"odinger equation for a particle in the braket notation: 
\begin{equation}
     \hat{H}\ket{\psi}=i\hbar \partial_t\ket{\psi}; \quad \ket{\psi(t_0)} = \ket{\psi_0}.
    \label{eq:schrodinger}
\end{equation}
Here the $\hat H$ is the Hamiltonian of the particle, $\ket{\psi}$ the state of the particle at time $t$, and $\ket{\psi_0}$ the initial condition at time $t_0$.  Let us define the temporal evolution operator by
\begin{equation}
    \hat U(t_0,t)\ket{\psi(t_0)} = \ket{\psi(t)}.
    \label{eq:q-odet}
\end{equation}
 Using Eq.~\ref{eq:schrodinger}, the evolution equation of $\hat{U}$ be written  as
\begin{equation}
    i\hbar\partial_t \hat U = \hat H \hat U; \quad \hat U(t_0) = \hat 1.
    \label{eq:evo-qodet}
\end{equation}
Since the particle is in a state $\ket{\psi(t_0)}$ at $t=t_0$, the operation of temporal evolution brings the particle to a new state $\ket{\psi(t)}$. Then, the probability of finding the particle in a position $x$ is calculated by a braket operation 
$\psi(t,x)=\braket{x|\psi(t)}$. This is the wave function of the particle, which can be calculated by multiplying Eq.~\ref{eq:q-odet} by $\bra{x}$ 
\begin{equation}
    \psi(t,x) = \bra{x}\hat U(t_0,t)\ket{\psi(t_0)}.
    \label{eq:wave-t}
\end{equation}
Let us introduce the completeness operator for the observable $\hat x$
\begin{equation}
    \hat 1 = \int{dx_0 \ket{x_0}\bra{x_0}}
    \label{eq:one}
\end{equation}
Inserting this operator into Eq.~\ref{eq:wave-t} we obtain
\begin{equation}
\begin{split}
    \psi(t,x) &= \bra{x}\hat U(t_0,t)\hat 1\ket{\psi(t_0)} \\
    &=\int{dx_0\bra{x}\hat U(t_0,t) \ket{x_0}\braket{x_0|\psi(t_0)}}\\
    &=\int{dx_0G(x,t,x_0,t_0)\psi(t_0,x_0)}.
    %\label{eq:GreenFunction}
\end{split}
 \label{eq:convolution_long}
\end{equation}
Here we define the {\it propagator} by the function
\begin{equation}
    G(x,t,x_0,t_0)=\bra{x}\hat U(t_0,t) \ket{x_0}
    \label{eq:propagator}
\end{equation}
According to Eq.~\ref{eq:convolution}, it is noted that the propagation behaves similarly to the Green functions of the partial differential equation.  This function can be calculated for some specific cases. If the Hamiltonian is independent of time, the evolution operator of Eq.~\ref{eq:evo-qodet} can be easily solved to obtain
\begin{equation}
    \hat{U}(t,t_0) = e^{-\frac{i}{\hbar}\hat{H} (t-t_0)}.
    \label{eq:sol-qodet}
\end{equation}
Eq.~\ref{eq:sol-qodet} can be replaced with Eq.~\ref{eq:propagator} to obtain the explicit expression of the propagator. This requires to introduce the monochromatic wave functions that result from the eigenstates of the  momentum observable $\hat p$
\begin{equation}
    \ket{p}=\frac{1}{\sqrt{2\pi}}\int{dx_0 e^{\frac{i}{\hbar} px_0\ket{x_0}}}
\end{equation}
The wave functions of the eigenstates $\ket{x}$  of the observables $\hat{x}$ is given by the Dirac's delta function
\begin{equation}
    \psi_{x_0}(x) =\braket{x|x_0}=\delta(x-x_0) 
\end{equation}
In agreement with Eq.~\ref{eq:one}, The property of the Dirac function hold for any function $f(x)$
\begin{equation}
    \int{dx f(x) \delta(x-x_0)}=f(x_0).
\end{equation}
The  corresponding wave function of the eigenstates $\ket{p_0}$ of the observable $\hat{p}$ is given by
\begin{equation}
    \psi_{p_0}(x) =\braket{x|p_0}=  \frac{1}{\sqrt{2\pi}}\int{dx_0 e^{\frac{i}{\hbar} px_0}\braket{x|x_0}}=\frac{1}{\sqrt{2\pi}}e^{\frac{i}{\hbar} p_0 x}
\end{equation}
Let us note that $i\hbar\partial_x \psi_{p_0}(x) = p\psi_{p_0}(x)$ which indicates that
\begin{equation}
    \bra{x}\hat p \ket{p_0}=-i\hbar\partial_x \braket{x|k}
    \label{eq:momentum}
\end{equation}
The completeness relation for the $\hat p$ observable also holds 
\begin{equation}
    \hat 1 = \int{dp_0 \ket{p_0}\bra{p_0}}
    \label{eq:one_p}
\end{equation}
Let us assume a Hamiltonian of the free particle
\begin{equation}
    \hat H = \hat H(p)
    \label{alpha H}
\end{equation}
Using Eq.~\ref{eq:one_p} in the propagator (Eq.~\ref{eq:propagator}) gives
\begin{equation}
\begin{split}
    &G(x,t,x_0,t_0) =\bra{x} e^{-\frac{i}{\hbar}\hat{H} (t-t_0)} \ket{x_0}\\
&=\int{dp \bra{x} e^{-\frac{i}{\hbar} \hat H{(p)} (t-t_0)} \ket{p}\braket{p|x_0}}\\
&=\int{dp e^{-\frac{i}{\hbar}H{(p)} (t-t_0)} \braket{x|p}\braket{p|x_0}}\\
&=\frac{1}{2\pi}\int{dp \exp{-\frac{i}{\hbar} H{(p)} (t-t_0)+\frac{i}{\hbar} p(x-x_0)}}
    \label{eq:propagator_alpha}
\end{split}
\end{equation}
Hence the propagator satisfies 
\begin{equation}
    G(x,t,x_0,t_0)=G(x-x_0,t-t_0).\nonumber
\end{equation}
 Using this property, and assuming that $t_0=0$, we can simplify the calculation of the wave function by Eq.~\ref{eq:convolution_long} to
\begin{equation}
    \psi(t,x) =\int_{-\infty}^{\infty}{G(x-x_0,t)\psi_0(x_0)dx_0}.
    \label{eq:convolution}
\end{equation}
Where $\psi_0(x)$ is the initial condition and 
\begin{equation}
    G(x,t)=\frac{1}{2\pi}\int_{-\infty}^{\infty}{e^{-\frac{i}{\hbar} H(p)t+i\hbar px}dp}
    \label{eq:propagator_simple_app}
\end{equation}

\section{Scaling Relations}~\label{app:scaling}
In this appendix, we focus on the scaling properties of the Salpeter equation, which helps us to recognize the general features of the model. Note that the full functionality of the green function is $G(x,t;x',t';m)$, but without loss of generality we set $x'=t'\equiv 0$, and $G(x,t,m)\equiv G(x,t;0,0;m)$, keeping in mind that eventually we should replace $x\to x-x'$ and $t\to t-t'$. We have the following integral expression
\begin{equation}
   G(x,t,m)\equiv \frac{1}{2\pi}e^{imt}\int_{0}^{\infty} e^{-i\sqrt{m^2+p^2}t}\cos{(px)}\text{d}p.
\end{equation}
One can easily show by inspection that for any positive scaling parameter $\lambda>0$ we have
\begin{equation}
G(\lambda x, \lambda t, m)=\lambda^{-1}G(x,t,\lambda m),
\label{Eq:scaling}
\end{equation}
so that under the transformation 
\begin{equation}
x\rightarrow \lambda x\ , \ t\rightarrow \lambda t \ , \ m\rightarrow \lambda^{-1} m
\label{Eq:scaling_trans}
\end{equation}
we have $G\rightarrow \lambda^{-1}G$. Therefore, we expect the following identity:
\begin{equation}
G(x,t)=m\tilde{G}\left[m(a x + bt)\right],
\end{equation}
where $\tilde{G}$ is a general function, that should be determined using other methods, and $a$ and $b$ are some constant complex numbers. The argument of $\tilde{G}$ guarantees the invariance of $\tilde{G}$ under scaling transformation. The limit $m\to 0$ helps to identify this function. To this end, we add a vanishingly small quantity to time, i.e. $t \to t-i0^+$ (which guarantees that the integral is well-defined in the UV limit $p\to \infty$), so that
\begin{equation}
G(x,t,m=0)=\frac{1}{2\pi}\left[\frac{i}{x-t+i0^+}-\frac{i}{x+t-i0^+}\right].
\end{equation}
Using the identity
\begin{equation}
\frac{1}{X+i 0^{\pm}}=\text{P.V.}\frac{1}{X}\mp i\pi \delta(X),
\end{equation}
where $\text{P.V.}$ stands for principal value, we obtain
\begin{equation}
G(x,t,m=0)=\frac{it}{\pi(x^2-t^2)}+\frac{1}{2}\delta(x-t)+\frac{1}{2}\delta(x+t),
\label{Eq:Imzero}
\end{equation}
where $\delta(x)$ is a Dirac delta function, and the first term on the right is $\text{P.V}$. One may compare this result with the orthogonality identity for plane waves
\begin{equation}
\frac{1}{2\pi}\int_{-\infty}^{\infty}e^{ip(x\pm t)}\text{d}p=\frac{1}{\pi}\lim_{p\to\infty}\frac{\sin p(x\pm t)}{x\pm t}=\delta(x\pm t),
\end{equation}
from which we realize that the first term in the r.h.s. of Eq.~\ref{Eq:Imzero} is due to the presence of \textit{absolute value} in the dispersion relation $E=\left|p\right|$ in the $m\to 0$ limit.
From this limit we realize that $a=1$ and $b=\pm 1$. We additionally find that for very small $m$ we have
\begin{equation}
\tilde{G}(x,t,\text{small}\ m)=\frac{1}{2\pi}\left[\frac{i}{m(x-t)+i0^+}-\frac{i}{m(x+t)-i0^+}\right].
\label{Eq:small}
\end{equation}
The other quantity of interest is $\left\langle x^2\right\rangle $. From the symmetry Eq.~\ref{Eq:scaling} one finds that
\begin{equation}
\left\langle x^2\right\rangle = m^{-2}\tilde{f}(mt)=t^2\tilde{g}(mt),
\end{equation}
where $\tilde{f}(y)=y^2\tilde{g}(y)$ are general smooth functions. For the massless free particle limit, we know
\begin{equation}
\left\langle x^2\right\rangle_{\text{massless}}=t^2,
\end{equation}
from which we see 
\begin{equation}
\left.\tilde{g}(y)\right|_{m\to 0}\to 1.
\end{equation}
\section{Integral Transformation Theorem}~\label{app:theorem}

\label{sec:Theorems}
Analytical integration of improper functions is available only in simple cases. If the integrand poses strong oscillatory behavior, numerical integration is also inconvenient and often lacks accuracy. Using the Cauchy theorem of complex variable analysis, it is possible to transform integrals with oscillatory behavior to integrands with smooth exponential decreases that are easier to integrate.  

\begin{theorem}
If $f(z)=e^{\Phi(z)}$ is is a function that is analytical in the first quarter of the complex plane:
\begin{equation}
    D_1=\{z=\rho e^{i\theta}=p+iq;0\le\theta\le\pi/2\}\nonumber 
\end{equation}
and $\Phi(z)$, with $z=\rho e^{i\theta}=p+iq$\, satisfies 
\begin{equation}
    \lim_{\rho\to\infty}{\Re{\Phi(\rho e^{i\theta})}}=-\infty 
    \quad z \in D_1
    \label{eq:condition_1}
\end{equation}
its integral over the semi-real axis $\mathbb{R}\cap D_n$ can be transformed into an integral over the semi-imaginary axis $\mathbb{I}\cap D_1$ as
\begin{equation}
\int_0^\infty{ e^{\Phi(p)} dp} 
=\int_0^\infty{ e^{\Phi(iq)} idq }.
\label{eq:complex_cond}
\end{equation} 

\end{theorem}

\begin{proof}
Since $f(z)$ is analytical in the first quadrant, its integral over any of the circular sectors on the first quadrant is zero
\begin{equation}
    \oint_{D_1}{ f(z)dz}=\oint_{D_1}{ e^{\Phi(z)} dz}=0
\end{equation}
Applying this property to a circular sector of infinite radius,
\begin{equation}
\int_{0}^\infty{ e^{\Phi(p)} dp}
+\lim_{\rho\to\infty}\int_{0}^{\frac{\pi}{2}}{ e^{\Phi(\rho e^{i\theta} )} d\theta}
+\int_{\infty}^{0}{ e^{\Phi(iq)} idq }=0\nonumber
\end{equation}
Thus,
\begin{equation}
\int_{0}^\infty{ e^{\Phi(p)} dp}
=\int_{0}^{\infty}{ e^{\Phi(iq)} idq }
-\lim_{\rho\to\infty}\int_{0}^{\frac{\pi}{2}}{ e^{\Phi(\rho e^{i\theta} )} d\theta}\nonumber
\end{equation}
The second term is expected to vanish under the condition Eq.~\ref{eq:complex_cond}
\end{proof}

Now we provide a similar theorem for the four quadrants of the complex plane.

\begin{theorem}
If $f(z)=e^{\Phi(z)}$ is is a function that is analytical in the four quarter of the complex plane:
\begin{equation}
    D_4=\{z=\rho e^{i\theta}=p+iq;-\pi/2\le\theta\le 0\}\nonumber 
\end{equation}
and $\Phi(z)$, with $z=\rho e^{i\theta}=p+iq$\, satisfies 
\begin{equation}
    \lim_{\rho\to\infty}{\Re{\Phi(\rho e^{i\theta})}}=-\infty 
    \quad z \in D_4
    \label{eq:condition_4}
\end{equation}
the integral over the semi-real axis $\mathbb{R}\cap D_4$ can be transformed into an integral over the semi-imaginary axis $\mathbb{I}\cap D_4$ as
\begin{equation}
\int_0^\infty{ e^{\Phi(p)} dp} =\int_0^\infty{ e^{\Phi(-iq)} (-i)dq }
\label{eq:complex_int2}
\end{equation} 
\end{theorem}

\begin{proof}
The theorem will be proof for the fourth quadrant. Since $f(z)$ is analytical in the four quadrants, its integral over any of the circular sectors on the first quadrant is zero
\begin{equation}
    \oint_{D_4}{ f(z)dz}=\oint_{D_4}{ e^{\Phi(z)} dz}=0\nonumber
\end{equation}
Applying this property to a circular sector of infinite radius,
\begin{equation}
\int_{0}^\infty{ e^{\Phi(p)} dp}
+\lim_{\rho\to\infty}\int_{0}^{-\frac{\pi}{2}}{ e^{\Phi(\rho e^{i\theta} )} d\theta}
+\int_{-\infty}^{0}{ e^{\Phi(iq)} idq }=0\nonumber
\end{equation}
Thus, replacing $q$ by $-q$ in the last integral,
\begin{equation}
\int_{0}^\infty{ e^{\Phi(p)} dp}
=\int_{0}^{\infty}{ e^{\Phi(-iq)} (-idq) }
-\lim_{\rho\to\infty}\int_{0}^{-\frac{\pi}{2}}{ e^{\Phi(\rho e^{i\theta} )} d\theta}\nonumber.
\end{equation}
The second term is expected to vanish under the condition \ref{eq:condition_4}
\end{proof}

\section{Theorem of Residues}
\label{app:theoremofResidues}

The main focus is to calculate the following integral 
\begin{equation}
    \psi^*(t,x) =\int_C{G(z,t)\psi_0(x-z)dz}.
\end{equation}
where $C$ is the contour in the complex plane shown in Fig.~\ref{fig:avanti}.
The function $G(z,t)$ is 
the analytical continuation of
the propagator (Eq.~\ref{eq:propagator_solve}) in the complex plane: 
\begin{equation}
    G(z,t)=\frac{1}{2\pi}\int_{-\infty}^{\infty}{e^{-iH_S(p)t+ipz}dp}
    \label{eq:propagator_solve_complex}
\end{equation}
and the function $\psi_0(z)$ is 
%the analytical continuation of
the initial wave function (Eq.~\ref{eq:wave_initial})
\begin{equation}
   \psi_0(z) = \sqrt{\frac{2}{\delta}}\cos{\frac{\pi z}{\delta}}(|z|<\delta/2). 
   \label{eq:wave_initial_complex}
\end{equation}
%%The first part of the demonstration (Fatima) is to use the Theorem of the Residues to probe that  
%%\begin{equation}
 %%    \oint_C{G(z,t)\psi_0(x-z)dz}=\frac{1}{2}\psi_0(x-t)+\frac{1}{2}\psi_0(x+t).
%%\label{eq:Residues_a}
%%\end{equation}
%%The second part of the demonstration (Morteza) is to show that the integrals over the infinitesimally large arc and the two infinitesimaly small arcs of Fig.~\ref{fig:avanti} are zero.  
%%\begin{equation}
%%\lim_{r\to\infty} \int_{C_r}{G(z,t)\psi_0(x-z)dz}=0
%%\label{eq:Residues_b}
%%\end{equation}
%%\begin{equation}
%%\lim_{\epsilon\to 0} \int_{C^i_\epsilon}{G(z,t)\psi_0(x-z)dz}=0\quad i=1,2
%%\label{eq:Residues_c}
%%\end{equation}
%%\textcolor{red}{By Fatima:}
In this appendix, we provide a proof of Eq.~\ref{eq:avantiwave}.
The integral~\ref{eq:avantiwave} diverges because of the non-integrable poles on the
integration contour. Since these poles are of the first order, the Cauchy principal value~$(PV)$ of the integral exists and corresponds to avoiding the poles with semi-circular arcs. 
To prove, we use the Cauchy residue theorem which states that if the integrand is analytic on a simply connected domain except for a finite number of isolated singularities, then the integral of that around a closed curve is equal to $2\pi i$ times the sum of the residues of the integrand at the singular points inside the closed curve.
First, we consider the case $m=0$ and
show that the~~$PV$ is: 
\begin{equation}
PV  \int_{-\infty}^{\infty}{G(y,t)\psi_0(x-y)dy}= \frac{1}{2}\left(\psi_0(x-t)+\psi_0(x+t) \right).
\label{eq:Residues_1}
\end{equation}
The contour in Fig.~\ref{fig:avanti} is chosen so that there are two poles inside it and so that the little circles around each of
the poles are so small. The contour is
\begin{equation}
C=C^{R}+C^{\epsilon}_{1}+C^{\epsilon}_{2}+C^{\infty}_{-\infty},
\label{eq:Residues_2}
\end{equation} 
where $C^{\infty}_{-\infty}$ is the integral in Eq.~\ref{eq:Residues_1}, $C^{\epsilon}_{i},\,i=1,2,$ are the contributions due to integrating over the poles at $\pm t$, and $C^{R}$
 is the (vanishing) contribution due to the large arc.
 \\
 To integrate around the pole $"-t "$, let us parameterize $z=-t+\epsilon e^{i\theta}$, then $dz=i\epsilon e^{i\theta}d\theta$. For the case $m=0,$ by considering Eq.~\ref{eq:masslesspropagator}, integrating
around the associated semi-circle yields
\begin{flalign}\nonumber
\int_{C^1_\epsilon}&{G(z,t)\psi_0(x-z)dz}\\
\nonumber
&=\lim_{\epsilon\to 0} \int_{\pi}^{2\pi}{G(-t+\epsilon e^{i\theta},t)\psi_0(x+t-\epsilon e^{i\theta})i\epsilon e^{i\theta}d\theta}\\
\nonumber
&=\lim_{\epsilon\to 0} \frac{-t}{\pi}\int_{\pi}^{2\pi}{
\dfrac{\psi_0(x+t-\epsilon e^{i\theta})}{\epsilon e^{i\theta}(-2t+\epsilon e^{i\theta})
} \epsilon e^{i\theta}d\theta}\\
&=\frac{1}{2} \psi_0(x+t).
\label{eq:Residues_3}
\end{flalign}
Similarly, to integrate around the pole $t$, we parameterize $z=t+\epsilon e^{i\theta}$, then $dz=i\epsilon e^{i\theta}d\theta$. By considering Eq. \ref{eq:masslesspropagator},
integrating
around the other semi-circle yields 
\begin{equation} 
\int_{C^2_\epsilon}{G(z,t)\psi_0(x-z)dz} 
=\frac{1}{2} \psi_0(x-t).
\label{eq:Residues_4}
\end{equation}
Also, the integration along the semi-circular portion $C^R$ of the contour $C$ vanishes
as the radius of integration becomes very large, i.e., as $R\to\infty,$ 
\begin{equation} 
\int_{C^R}{G(z,t)\psi_0(x-z)dz}=0,
\label{eq:Residues_5}
\end{equation}
and by considering the counter~\ref{eq:Residues_2}, this completes the proof of Eq.~\ref{eq:avantiwave}.\\
To prove Eq.~\ref{eq:Residues_1}, since there are two poles inside of the contour, we have by Cauchy’s residue theorem, 
\begin{equation} 
\int_{C}{G(z,t)\psi_0(x-z)dz}=2i\pi \left(Res(-t)+Res(t)\right).
\label{eq:Residues_6}
\end{equation}
By using Theorem~IV(p.~557) from~\cite{hildebrand1962advanced} and Equations~\ref{eq:Residues_3}, and \ref{eq:Residues_4}, 
\begin{equation} 
 Res(\pm t)=\frac{1}{2i\pi} \psi_0(x\mp t),
\label{eq:Residues_4_1}
\end{equation}
and by considering the counter~\ref{eq:Residues_2}, this completes the proof of Eq.~\ref{eq:Residues_1}.\\
Now, we consider the case $m\neq0.$ We follow a similar scenario on the same contour. Therefore, under these conditions, Equations \ref{eq:Residues_2}, \ref{eq:Residues_5}, and \ref{eq:Residues_6} work for our case and we should investigate Eq. \ref{eq:Residues_1}.
 To achieve this, we apply the analytical solution of the Salpeter propagator, as described in Eq. \ref{eq:sal_ana_summary}, for integrals   \ref{eq:Residues_3} and \ref{eq:Residues_4}. 
We emphasize the case $m\neq0$ by denoting
\begin{equation}
G^m(z,t)  =
\begin{cases}
    \frac{imt e^{imt}}{\pi\sqrt{z^2-t^2}}K_1\left(m\sqrt{z^2-t^2}\right), & \vert z\vert>t\\
    \\
    -\frac{mt e^{imt}}{2\sqrt{t^2-z^2}}H_1^{(1)}\left(-m\sqrt{t^2-z^2}\right), & \vert z\vert<t,
\end{cases}
\end{equation}
where $H_1^{(1)}$ is the Hankel function of first kind, and $K_1$ is the modified Bessel function of the second kind. 
\\
Integrating
around the first semi-circle yields
\begin{flalign}\nonumber
\int_{C^1_\epsilon}&{G^m(z,t)\psi_0(x-z)dz}\\
\nonumber
&=\lim_{\epsilon\to 0} \int_{\pi}^{2\pi}{G^m(-t+\epsilon e^{i\theta},t)\psi_0(x+t-\epsilon e^{i\theta})i\epsilon e^{i\theta}d\theta}\\
\nonumber
&=\lim_{\epsilon\to 0} \frac{-imte^{imt}}{\pi}\int_{\pi}^{2\pi} 
\frac{K_1\left(m \epsilon^{1/2}e^{i\theta/2}\sqrt{-2t+\epsilon e^{i\theta}}\right)}{ \epsilon^{1/2}e^{i\theta/2}\sqrt{-2t+\epsilon e^{i\theta}}} ...
\\
\nonumber
&\,\,\,\,\,\,\,\,\,\,\,\,\,\,\,\,\,\,\,\,\,\,\,\,\,\,\,\,\,\,\,\,\,\,\,\,\,\,\,\,\,\,\,\,\,\,\,\,\,\,\,\,...\times\psi_0(x+t-\epsilon e^{i\theta})i\epsilon e^{i\theta}d\theta
\label{eq:Residues_7}
\end{flalign}
Applying the fact
$K_1(y) \to \frac{1}{y}$ as $y \to o,$ we get
\begin{equation} 
\int_{C^\epsilon_1}{G^m(z,t)\psi_0(x-z)dz}=\frac{1}{2}\psi_0(x+t)  e^{imt} .
\label{eq:Residues_8}
\end{equation}
Similarly, integrating
around the second semi-circle yields
\begin{equation} 
\int_{C^\epsilon_2}{G^m(z,t)\psi_0(x-z)dz}=\frac{1}{2}\psi_0(x-t)  e^{imt} .
\label{eq:Residues_9}
\end{equation}
We can easily acquire  
\begin{equation}
 Res(\pm t)=\frac{e^{imt}}{2i\pi} \psi_0(x\mp t),
\label{eq:Residues_9_1}
\end{equation}
then, the $PV$ is:
\begin{flalign}\nonumber
PV  \int_{-\infty}^{\infty}&{G(y,t)\psi_0(x-y)dy}\\
&= \frac{1}{2}\left(\psi_0(x-t)+\psi_0(x+t) \right)e^{imt}.
\label{eq:Residues_10}
\end{flalign}
\section{Analytical expressions based on perturbative expansion}~\label{app:expansion}
In this section we present an analytical expansion of the salpeter integral, which serves as a perturbative expansion. It helps much to realize the analytical form of the function around e.g. the light cone, and other regions of interest, and also shows the internal structure and consistency of the analytical expressions. We follow the expansion scheme presented according to an \textit{ad hoc} weight function. In this method we multiply the integrand by a weight function $f(p,p_0)$ which makes the approximate integration feasible in an appropriate limit, and $p_0$ is a characteristic cutoff momentum dividing the integral domain to two intervals, each being treated differently in the perturbation expansion. In our case, $p_0\equiv m$ is the characteristic momentum that separates ``large'' and ``short'' $p$ values. The choice of $f(p,p_0)$ is not unique, and can be any form depending on the form of the integral, for example it can be either ($p_0\equiv m$) 
\begin{equation}
\begin{split}
f(p,p_0)&=\exp\left[-\left( \frac{p}{p_0}\right)^2\right]\\
&\text{or}\\
f(p,p_0)&=\Theta(p-p_0),
\end{split}
\end{equation}
where $\Theta(x)$ is a step function. In this paper choose the second ``cut'' function, and $p_0\equiv m$. We then cast the formula Eq.~\ref{eq:propagator_simple} to the following:
\begin{equation}
\begin{split}
    G(x,t,0,0)=&\int_0^{\infty}\text{d}p\left[f(p,p_0)+1-f(p,p_0)\right]\times\\
    &\exp [-itH(p)]\cos\left(px\right)\\
    &\equiv G_1(x,t)+G_2(x,t).
    \end{split}
\end{equation}
The integrand of function 
\begin{equation}
    G_1(x,t)=\int_0^{\infty}\text{d}p f(p,p_0)\exp [-itH(p)]\cos\left(px\right)
\end{equation}
can be expanded for small $p$ values ($p\ll p_0$), while \begin{equation}
    G_2(x,t)=\int_0^{\infty}\text{d}p \left[1-f(p,p_0)\right]\exp [-itH(p)]\cos\left(px\right)
\end{equation}
should be expanded around small $\frac{1}{p}\ll\frac{1}{p_0}$. We end up to ($\epsilon_1\equiv \frac{p}{m}$, and $g_1(mt;\epsilon_1)\equiv \exp\left[-imt \left(\sqrt{1+\epsilon_1^2}\right]-1\right)$, $\rho\equiv mx$)
\begin{equation}
\begin{split}
G_1(x,t)&=m\int_0^1 \text{d}\epsilon_1 e^{-imt}g_1(mt;\epsilon_1)\cos (\epsilon_1 mx)\\
&= m\sum_{n=0}^{\infty}e^{-imt}\frac{g_1^{(n)}\left(mt\right)}{n!}\int_0^1\text{d}\epsilon_1\epsilon_1^n \cos (\epsilon_1 mx)\\
&= me^{-imt}\sum_{n=0}^{\infty}\frac{g_1^{(n)}\left(mt\right)}{n!}\text{Re}\left(-i\frac{\partial}{\partial\rho}\right)^n\int_0^1\text{d}\epsilon_1 e^{i\rho\epsilon_1}\\
&= me^{-imt}\sum_{n=0}^{\infty}\frac{g_1^{(n)}\left(mt\right)}{n!}f_n(mx)
\end{split}
\end{equation}
where 
\begin{equation}
g_1^{(n)}(mt)\equiv \frac{\partial^n}{\partial \epsilon_1^n}\left. g_1(mt;\epsilon_1)\right|_{\epsilon_1=0},
\label{Eq:g1}
\end{equation}
and
\begin{equation}
\begin{split}
f_n(\rho)&\equiv \text{Re}\left[\left(-i\frac{\partial}{\partial\rho}\right)^n\left(\frac{e^{i\rho}-1}{i\rho}\right)\right]\\
&=\left\lbrace \begin{matrix}
(-1)^{\frac{n+1}{2}}\left(\frac{\partial}{\partial\rho}\right)^n\left(\frac{\cos \rho-1}{\rho}\right) & n=\text{odd}\\
(-1)^{\frac{n}{2}}\left(\frac{\partial}{\partial\rho}\right)^n\left(\frac{\sin \rho}{\rho}\right) & n=\text{even}
    \end{matrix}\right. .
\end{split}
\label{Eq:fnro}
\end{equation}
Note that $g_1^{(n)}=0$ for all \textit{odd} $n$ values, while for the even ones we have
\begin{equation}
\begin{split}
&g_1^{(0)}(y)=1,\\
&g_1^{(2)}(y)=-iy,\\
&g_1^{(4)}(y)=3y(i-y),\\
&g_1^{(6)}(y)=15y(-3i+3y+iy^2),\\
&g_1^{(8)}(y)=15y(15i+15y-6iy^2+y^3),\\
&...
\end{split}
\label{Eq:AB1}
\end{equation}
The asymptotic behavior of $g_1^{(2n)}$ is
\begin{equation}
  g_1^{(2n)}(y)\to\left\lbrace \begin{matrix}
   A_n^{(1)}y^n &  y\gg 1\\
   B_n^{(1)}y & y\ll 1
\end{matrix}\right. 
\end{equation}
where $A_n^{(1)}$ and $B_n^{(1)}$ are two expansion coefficients that are given in Eq.~\ref{Eq:AB1}.\\
$f_n(\rho)$ can also be written in terms of generalized hypergeometric function as follows
\begin{equation}
f_{n}(\rho)=\frac{{_1F_2}\left[\frac{n+1}{2},\left\lbrace\frac{1}{2},\frac{n+3}{2}\right\rbrace,-\frac{\rho^2}{4} \right]}{n+1}.
\label{Eq:fHyper}
\end{equation}
The explicit form of $f_n$ is
\begin{equation}
\begin{split}
&f_0(\rho)=\frac{\sin\rho}{\rho}\\
&f_1(\rho)=\frac{\cos \rho-1+\rho\sin\rho}{\rho^2}\\
&f_2(\rho)=\frac{2\rho\cos \rho+(\rho^2-2)\sin\rho}{\rho^3}\\
&f_3(\rho)=\frac{6+3(\rho^2-2)\cos \rho+\rho(\rho^2-6)\sin\rho}{\rho^4}\\
&f_4(\rho)=\frac{4\rho(\rho^2-6)\cos \rho+(24-12\rho^2+\rho^4)\sin\rho}{\rho^5}\\
&...
\end{split}
\end{equation}
Note that in the limit $\rho\to 0$ we have
\begin{equation}
\begin{split}
 f_n(\rho)&=\text{Re}\left[\left(-i\frac{\partial}{\partial \rho}\right)^n\sum_{m=1}^{\infty}\frac{1}{m!}(i\rho)^{m-1}\right]\\
 &=\text{Re}\left[\frac{n!}{(n+1)!}+i\frac{(n+1)!}{(n+2)!}\rho+\frac{(n+2)...4\times 3}{(n+3)!}\rho^2 + O(\rho^3)\right]\\
 &=\frac{1}{n+1}-\frac{\rho^2}{2(n+3)}+O(\rho^3).
\end{split}
\label{Eq:smallf}
\end{equation}
For the large $\rho$ limit one can see by inspection that 
\begin{equation}
\begin{split}
\left. f_n(\rho)\right|_{\rho\gg 1}\to a_n\frac{\sin\rho}{\rho}
\end{split}
\label{Eq:largef}
\end{equation}
where $a_n=1$ for all $n$ values. \\
For $G_2$ we have:
\begin{equation}
\begin{split}
G_2(x,t)&=\int_m^{\infty}\text{d}pe^{-ipt\sqrt{1+\left(\frac{m}{p}\right)^2}}\cos px\\
&=\int_m^{\infty}\text{d}pg_2(mt,\epsilon_2)e^{-ipt}\cos px\\
&=m\sum_{n=0}^{\infty}\frac{g_2^{(n)}(mt)}{n!}\int_1^{\infty}\text{d}\tilde{p}\tilde{p}^{-n}e^{-i\tilde{p}mt}\cos (\tilde{p}mx)
\end{split}
\label{Eq:g2}
\end{equation}
where $\tilde{p}\equiv p/m$, $\epsilon_2\equiv \frac{1}{\epsilon_1}= \frac{m}{p}$, and
\begin{equation}
\begin{split}
g_2(mt,\epsilon_2)&\equiv \exp\left[-i\frac{mt}{\epsilon_2}\left(\sqrt{1+\epsilon_2^2}-1\right)\right]\\
&=\sum_{n=0}^{\infty}\frac{1}{n!}g_2^{(n)}(mt)\epsilon_2^n=\sum_{n=0}^{\infty}\frac{1}{n!}g_2^{(n)}(mt)\tilde{p}^{-n},
\end{split}
\end{equation}
and 
\begin{equation}
g_2^{(n)}(mt)\equiv \left. \frac{\partial^n}{\partial \epsilon_2^n}g_2(mt,\epsilon_2)\right|_{\epsilon_2=0}
\end{equation}
is the coefficients of the Taylor expansion. Notice that when $\epsilon_2\rightarrow 0$, $\epsilon_2^{-1}\left(\sqrt{1+\epsilon_2^2}-1\right)\rightarrow 0$. \\
The coefficients are easily obtained by inspection as follows:
\begin{equation}
\begin{split}
&g_2^{(0)}(y)=(-i)^0,\\
&g_2^{(1)}(y)=(-i)^1\frac{y}{2},\\
&g_2^{(2)}(y)=(-i)^2\frac{y^2}{2^2},\\
&g_2^{(3)}(y)=(-i)^3\frac{y}{2^3}(3!+y^2),\\
&g_2^{(4)}(y)=(-i)^4\frac{y^2}{2^4}(4!+y^2),\\
&g_2^{(5)}(y)=(-i)^5\frac{y}{2^5}(2\times 5!+\frac{1}{2}5!y^2+y^4),\\
&g_2^{(6)}(y)=(-i)^6\frac{y^2}{2^6}(\frac{5}{2}\times 6!+5!y^2+y^4),\\
&g_2^{(7)}(y)=(-i)^7\frac{y}{2^7}(5\times 7!+\frac{3}{2}7!y^2+\frac{7}{4}5!y^4+y^6),\\
&g_2^{(8)}(y)=(-i)^8\frac{y^2}{2^8}(7\times 8!+\frac{14}{3}7!y^2+\frac{7}{15}6!y^4+y^6)...
\end{split}
\label{Eq:g2Expansion}
\end{equation}
\newline
\par
Note that 
\begin{equation}
  g_2^{(n)}\to\left\lbrace \begin{matrix}
   (-i)^n\frac{y^n}{2^n} &  y\gg 1\\
   (-i)^n\frac{n!}{2^n}(A_n^{(2)})y^{m_n} & y\ll 1
\end{matrix}\right.
\end{equation}
where $m_n=2$ and $1$ for $n$ even and odd respectively, and $A_n^{(2)}$ is an expansion coefficient given in Eq.~\ref{Eq:g2Expansion}.
To process the Eq.~\ref{Eq:g2} we use the fact that
\begin{equation}
\begin{split}
I_n(mt,mx)&\equiv\int_1^{\infty}\text{d}\tilde{p}\tilde{p}^{-n}e^{-i\tilde{p}mt}\cos (\tilde{p}mx)\\
&=\frac{1}{2}\text{Ei}_n\left[im(t-x)\right]+\frac{1}{2}\text{Ei}_n\left[im(t+x)\right],
\end{split}
\end{equation}
where $\text{Ei}_n(y)\equiv \int_1^{\infty}z^{-n}e^{-yz}\text{d}z$ is the exponential integral function. Therefore, we find that
\begin{equation}
\begin{split}
G_2(x,t)=\frac{m}{2}\sum_{n=0}^{\infty}\frac{g_2^{(n)}(mt)}{n!}&\left(\text{Ei}_n\left[im(t-x)\right]\right.\\
& \left. +\text{Ei}_n\left[im(t+x)\right]\right)
\end{split}
\label{Eq:g2_final}
\end{equation}
Note that the singularities of the propagator are due to the presence of the exponential integral functions in $G_2(x,t)$. Given that 
\begin{equation}
\begin{split}
\text{Ei}_0(z) &=\frac{e^{-z}}{z},\\
\lim_{z\to 0} \text{Ei}_1(z) &=-\gamma-\ln z,\\
\lim_{z\to 0} \text{Ei}_n(z) &= \frac{1}{n-1} \ \ , \ n>1   
\end{split}
\end{equation}
where $\gamma$ is an Euler–Mascheroni constant, we find that in the limit $x\to t$, $G_2$ becomes singular according to (the same happens for $x=-t$)
\begin{equation}
G_2(x,t)\to \frac{1}{2i(t-x)}+\frac{im^2t}{4}\ln [im(t-x)].
\label{Eq:singularity}
\end{equation}
\par
Finally, we find
\begin{widetext}
\begin{equation}
 G(x,t;0,0)= m\sum_{n=0}^{\infty}\frac{g_1^{(n)}\left(mt\right)}{n!(n+1)} {_1F_2}\left[\frac{n+1}{2},\left\lbrace\frac{1}{2},\frac{n+3}{2}\right\rbrace,-\frac{m^2x^2}{4} \right]+  \frac{m}{2}\sum_{n=0}^{\infty}\frac{g_2^{(n)}(mt)}{n!}\left(\text{Ei}_n\left[im(t-x)\right]+\text{Ei}_n\left[im(t+x)\right]\right).
\end{equation}
\end{widetext}
This expansion is singular at the Dirac cone, i.e. $x=t$ and $x=-t$ due to the presence of the exponential integral functions. More precisely, from Eq.~\ref{Eq:singularity} we have:
\begin{equation}
\begin{split}
\left. G(x,t;0,0)\right|_{x\to \pm t}& =\frac{1}{2i(t\mp x)}+\frac{im^2t}{4}\ln [im(t\mp x)]\\
&=\frac{1}{2z}+\frac{m^2z}{8}\ln (mz),
\end{split}
\label{Eq:divergence}
\end{equation}
where we have defined $z\equiv \tau +ix$, $\bar{z}\equiv \tau -ix$, and $\tau\equiv it$ is the Wick-rotated (imaginary) time. In the above equation we used $z=\tau+ix\approx 2it$ in the limit $x\to t$ (the same holds for $x\to -t$ with $z$ replaced by $\bar{z}$). All in all we find that
\begin{equation}
\begin{split}
&\left. G(x,t;0,0)\right|_{x\to t}+\left.G(x,t;0,0)\right|_{x\to -t} =\\
&G(x,t,m=0)+\frac{m^2}{8}\left[z\ln (mz)+\bar{z}\ln (m\bar{z})\right]. 
\end{split}
\end{equation}
In Appendix~\ref{app:scaling} we discuss the scaling properties of the propagator. Especially we show that the integral is proportional to the modified Bessel functions of the second type. More precisely,
\begin{equation}
\begin{split}
G(x,t) & =\frac{imt e^{imt}}{\pi\sqrt{x^2-t^2}}K_1\left(m\sqrt{x^2-t^2}\right),\ x^2\ge t^2\\
& = -\frac{mt e^{imt}}{2\sqrt{t^2-x^2}}H_1^{(1)}\left(-m\sqrt{t^2-x^2}\right),\ x^2<t^2,
\end{split}
\label{Eq:MainApp}
\end{equation}
where $K_1(y)$ is the modified Bessel function of second type, and $H_1^{(1)}$ is a Hankel function of first type. This can also be written in terms of the imaginary time $\tau$ as follows
\begin{equation}
\begin{split}
G(x,\tau) & =\frac{m}{\pi}\cos\theta e^{mr\cos\theta} K_1\left(mr\right),\  x^2\ge t^2\\
& = \frac{im}{2}\cos\theta e^{mr'\cos\theta}H_1^{(1)}\left(mr'\right), \ x^2<t^2,
\end{split}
\label{Eq:Main1}
\end{equation}
where $r\equiv \sqrt{x^2+\tau^2}$, $r'\equiv \sqrt{t^2-x^2}=ir$, and $\theta\equiv \tan^{-1}\left(\frac{x}{\tau}\right)$. For the stochastic process (SP) applications, one may use  the first branch, as follows
\begin{equation}
G_{\text{SP}}(r,\theta)=\frac{m}{\pi}\cos\theta e^{mr\cos\theta} K_1\left(mr\right).
\label{Eq:SP}
\end{equation}
This function has been shown in Fig.~\ref{fig:Plot1}
\begin{figure}
\includegraphics[width=0.9\linewidth]{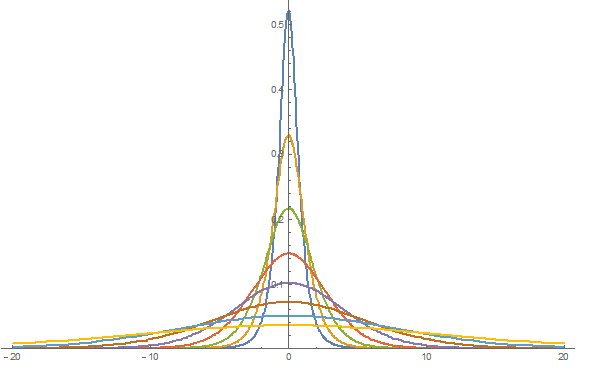}
\caption{A plot for Eq.~\ref{Eq:SP} for $\tau=2^n$, $n=0,1,2,3,4,5,6,7$ (from top to bottom).}
\label{fig:Plot1}	
\end{figure}

\section{Differential Equation}~\label{app:differential}
Here we present the connection of the Salpeter equation to the Klein-Gordon differential equation. By differentiating the Eq.~\ref{eq:propagator_solve} with respect to $x$ and $t$ One easily finds ($x,t\ne 0$)
\begin{equation}
\left(-\partial_x^2+\partial_t^2+m^2 \right)G\left(x,t,m\right)=0,
\end{equation}
which can be understood using the fact that $H\equiv i\partial_t$ and $p\equiv -i\partial_x$. Above, we considered a complex time $t+i0^+$ for control the integral limits. In an extreme situation, one may Wick-rotate the time, so that $\tau=it$ describes a new \textit{coordination} which is treated like $x$ as follows ($x,\tau\ne 0$)
\begin{equation}
\left(\nabla^2-m^2 \right)G^{(\tau)}=0,
\end{equation}
which is a homogeneous Helmholtz equation in two dimensions for $x,\tau\ne 0$, and $G^{(\tau)}(x,\tau,m)\equiv G(x,-i\tau,m)$, and $\nabla^2\equiv \partial_x^2+\partial_{\tau}^2$. Note that this equation respects the scaling relation given in Eq.\ref{Eq:scaling}. The resulting integral after the Wick rotation describes also the relativistic probability density of stochastic processes. In terms of this quantity, we have 
\begin{equation}
G^{(\tau)}(x,\tau,m=0)=\frac{\tau}{\pi(x^2+\tau^2)}+\frac{1}{2}\delta(x-i\tau)+\frac{1}{2}\delta(x+i\tau),
\label{Eq:Imzero2}
\end{equation}
i.e. when $x$ is real, only the first term remains, which is the Cauchy distribution. Using the identity
\begin{equation}
\lim_{\epsilon\to 0}\frac{1}{\pi}\frac{\epsilon}{x^2+\epsilon^2}=\delta(x),
\end{equation}
we find that 
\begin{equation}
\begin{split}
\nabla^2G^{(\tau)} (x,\tau,m=0) &=-\frac{8\tau}{x^2+\tau^2} \lim_{\epsilon\to 0}\frac{1}{\pi}\frac{\epsilon^2}{(x^2+\tau^2+\epsilon^2)^2}\\
&=-4 f(x,\tau)\delta(x)\delta(\tau),  
\end{split}
\end{equation}
where 
\begin{equation}
\begin{split}
f(x,\tau)\equiv &\frac{2 \tau}{(x^2+\tau^2)}=\frac{1}{\tau+ix}+\frac{1}{\tau-ix}\\
&=2\pi G^{(\tau)}(x,\tau,m=0)
\end{split}
\end{equation}
We immediately conclude that the same happens for general $m$
\begin{equation}
\begin{split}
\left(\nabla^2-m^2\right)G^{(\tau)}(x,\tau,m) =-4 f(x,\tau)\delta(x)\delta(\tau)
\end{split}
\label{Eq:master}
\end{equation}
Defining $r\equiv\sqrt{x^2+\tau^2}$ and $\theta\equiv \arctan\frac{x}{\tau}$, we rewrite it as ($r,\theta\ne 0$)
\begin{equation}
G_{rr}^{(\tau)}+\frac{1}{r}G_r^{(\tau)}+\frac{1}{r^2}G_{\theta\theta}^{(\tau)}-m^2G^{(\tau)}=0,
\label{Eq:rad_master}
\end{equation}
where $G_{rr}^{(\tau)}$ means double derivative with respect to $r$ and etc. Then, using the method of separation of variables $G^{(\tau)}\equiv R(r)\Theta(\theta)$, we find
\begin{equation}
\begin{split}
&\Theta''+n^2\Theta=0\\
&r^2R''+rR'-\left((mr)^2+n^2\right)R=0,
\end{split}
\end{equation}
from which one finds 
\begin{equation}
\Theta=a_n\sin(n\theta)+b_n\cos(n\theta)
\end{equation}
and $R$ is a modified Bessel function of the second kind
\begin{equation}
R(r)=K_n(mr)=\int_{0}^{\infty}\exp[-mr\cosh t]\cosh (nt)\text{d}t.
\end{equation}
Note that we have two independent radial solutions and the only one which is infinity at $r=0$ and regular at $r\to \infty$ is this term. We find eventually
\begin{equation}
G^{(\tau)} =\frac{1}{2\pi} \sum_nK_n(mr)\left(a_n\sin(n\theta)+b_n\cos(n\theta)\right)
\end{equation}
where $a_n$ and $b_n$ should be found using Eq.~\ref{Eq:Imzero2} for $m=0$. Given that 
\begin{equation}
\begin{split}
& K_n(y)\to \sqrt{\frac{\pi}{2}}\frac{\exp(-y)}{\sqrt{y}} , \ \ y\to \infty\\
& K_0(y)\to -\gamma -\ln \frac{y}{2} + O(y^2) \ \ y\to 0\\
&K_1(y)\to \frac{1}{y}+ \frac{y}{2}\ln \frac{y}{2} + O(y), \ \ y\to 0\\
&K_n(y)\to (-1)^{n-1}\left(\frac{y}{2}\right)^n\ln \frac{y}{2}\left(\frac{1}{n!}+O(y^2)\right), \ n\ge 2, y\to 0.
\end{split}
\label{Eq:asymptotic}
\end{equation}
Comparing this result with Eq.~\ref{Eq:divergence}, we notice that the singularity of the integral for $r\to 0$ (equivalent to $x=\pm t$) is by $1/r$ and $r\ln r$, so that the only term that contributes in the summation is $n=1$. More precisely, Eq.~\ref{Eq:divergence} gives us 
\begin{equation}
\begin{split}
G\to &\frac{m}{2}\left[\frac{1}{mr}e^{-i\theta}+\frac{1}{2}\frac{mre^{i\theta}}{2}\ln(mr)\right]\\
+ & \frac{m}{2}\left[\frac{1}{mr}e^{i\theta}+\frac{1}{2}\frac{mre^{-i\theta}}{2}\ln(mr)\right]\\
&=m\left[\frac{1}{mr}+\frac{1}{2}\frac{mr}{2}\ln(mr)\right]\cos\theta
\end{split}
\end{equation}
Then, we find that $a_1=m$, and $b_1=0$, and all the other $a_n,b_n$ ($n\ne 1$) are zero. The following integral identity confirms the above result 
\begin{equation}
\begin{split}
    \int_{0}^{\infty} &{dx \,\, \text{exp}(-\alpha \sqrt{x^2+\beta^2})\cos\gamma x } \\
    &= \frac{\alpha \beta}{\sqrt{\alpha^2+\gamma^2}}K_1(\beta\sqrt{\alpha^2+\gamma^2}), \text{Re}[\alpha]>0, \text{Re}[\beta]>0.
    \label{Eq:AnalyticA}
\end{split}
\end{equation}
\\
Using this, we find that
\begin{equation}
\begin{split}
G(x,t) & =\frac{imt e^{imt}}{\pi\sqrt{x^2-t^2}}K_1\left(m\sqrt{x^2-t^2}\right),\ x^2>t^2\\
& = -\frac{mt e^{imt}}{2\sqrt{t^2-x^2}}H_1^{(1)}\left(m\sqrt{t^2-x^2}\right),\ x^2<t^2,
\end{split}
\label{Eq:MainApp_1}
\end{equation}
where $H_1^{(1)}$ is a Hankel function of first kind.
Using new complex variables $z\equiv \tau+ix$ and $\bar{z}\equiv \tau-ix$, we can simplify the notation. First we note that $4\partial\bar{\partial}=\partial_{\tau}^2+\partial_x^2$, where $\partial \equiv \partial_z$, and $\bar{\partial}\equiv\partial_{\bar{z}}$. Equation~\ref{Eq:master} gives rise to
\begin{equation}
\left(4\partial\bar{\partial}-m^2\right)G^{(\tau)}(z,\bar{z},m)=-2 \left(\frac{1}{z}+\frac{1}{\bar{z}}\right)\delta^2(z)
\end{equation}
where $\delta^2(z)\equiv \delta(x)\delta(t)$, and we have used $\partial \frac{1}{\bar{z}}=\bar{\partial}\frac{1}{z}=\delta^2(z)$. For the $m=0$ case, $G^{(\tau)}$ is a summation of the holomorphic and anti-holomorphic terms, i.e.
\begin{equation}
G^{(\tau)}(z,\bar{z},m=0)=G_1(z)+G_2(\bar{z})=\frac{1}{\pi}\left[\frac{1}{z}+\frac{1}{\bar{z}}\right]
\end{equation}
When $m\ne 0$ however, it is not decomposable to these two parts. In this case, Eq.~\ref{Eq:AnalyticA} tells us that
\begin{equation}
G^{(\tau)}(z,\bar{z},m)\propto\frac{1}{\sqrt{z\bar{z}}}K_1(m\sqrt{z\bar{z}}).
\end{equation}

\end{document}